\documentclass{amsart}
\usepackage{hyperref}
\usepackage{graphicx}
\usepackage{curves}
\usepackage{enumerate}
\newcommand*{\mailto}[1]{\href{mailto:#1}{\nolinkurl{#1}}}

\newtheorem{theorem}{Theorem}[section]
\newtheorem{lemma}[theorem]{Lemma}

\newcommand{\R}{\mathbb{R}}
\newcommand{\Z}{\mathbb{Z}}
\newcommand{\N}{\mathbb{N}}
\newcommand{\C}{\mathbb{C}}
\newcommand{\T}{\mathbb{T}}

\newcommand{\nn}{\nonumber}
\newcommand{\beq}{\begin{equation}}
\newcommand{\eeq}{\end{equation}}
\newcommand{\bea}{\begin{eqnarray}}
\newcommand{\eea}{\end{eqnarray}}

\newcommand{\ol}{\overline}
\newcommand{\pa}{\partial}
\newcommand{\ti}{\tilde}
\newcommand{\wti}{\widetilde}

\newcommand{\I}{\mathrm{i}}
\newcommand{\E}{\mathrm{e}}

\newcommand{\re}{\mathop{\mathrm{Re}}}
\newcommand{\im}{\mathop{\mathrm{Im}}}
\newcommand{\sech}{\mathop{\mathrm{sech}}}
\DeclareMathOperator{\res}{Res}
\newcommand{\db}{\mathfrak{D}}

\newcommand{\eps}{\varepsilon}

\newcommand{\ga}{\gamma}
\newcommand{\om}{\omega}

\numberwithin{equation}{section}

\newcommand{\sigI}{\begin{pmatrix} 0 & 1 \\ 1 & 0 \end{pmatrix}}

\begin{document}

\title[Long-Time Asymptotics for the KdV Equation]{Long-Time Asymptotics for the Korteweg--de Vries Equation
with Steplike Initial Data}

\author[I. Egorova]{Iryna Egorova}
\address{Institute for Low Temperature Physics\\ 47,Lenin ave\\ 61103 Kharkiv\\ Ukraine}
\email{\href{mailto:iraegorova@gmail.com}{iraegorova@gmail.com}}

\author[Z. Gladka]{Zoya Gladka}
\address{Institute for Low Temperature Physics\\ 47,Lenin ave\\ 61103 Kharkiv\\ Ukraine}
\email{\href{mailto:gladkazoya@gmail.com}{gladkazoya@gmail.com}}

\author[V. Kotlyarov]{Volodymyr Kotlyarov}
\address{Institute for Low Temperature Physics\\ 47,Lenin ave\\ 61103 Kharkiv\\ Ukraine}
\email{\href{mailto:kotlyarov@ilt.kharkov.ua}{kotlyarov@ilt.kharkov.ua}}

\author[G. Teschl]{Gerald Teschl}
\address{Faculty of Mathematics\\ University of Vienna\\
Nordbergstrasse 15\\ 1090 Wien\\ Austria\\ and International Erwin Schr\"odinger
Institute for Mathematical Physics\\ Boltzmanngasse 9\\ 1090 Wien\\ Austria}
\email{\href{mailto:Gerald.Teschl@univie.ac.at}{Gerald.Teschl@univie.ac.at}}
\urladdr{\href{http://www.mat.univie.ac.at/~gerald/}{http://www.mat.univie.ac.at/\string~gerald/}}

\keywords{Riemann--Hilbert problem, KdV equation, steplike}
\subjclass[2000]{Primary 37K40, 35Q53; Secondary 37K45, 35Q15}
\thanks{Nonlinearity {\bf 26}, 1839--1864 (2013)}
\thanks{Research conducted in the framework of the project "Ukrainian branch of the French-
Russian Poncelet laboratory" --- "Probability problems on groups and spectral theory".
Research supported by the Austrian Science Fund (FWF) under Grant No.\ Y330.}

\begin{abstract}
We apply the method of nonlinear steepest descent to compute the long-time
asymptotics of the Korteweg--de Vries equation with steplike initial data.
\end{abstract}

\maketitle

\section{Introduction}

We study the long-time asymptotic behavior of solutions of the Korteweg--de Vries (KdV) equation
\beq\label{kdv}
q_t(x,t)=6q(x,t)q_x(x,t)-q_{xxx}(x,t), \quad (x,t)\in\R\times\R,
\eeq
with steplike initial data $q(x,0)=q(x)\in C^{11}(\R)$ such that
\beq \label{ini}\left\{ \begin{array}{ll} q(x)\to 0,& \ \ \mbox{as}\ \ x\to +\infty,\\
q(x)\to -c^2,&\ \ \mbox{as}\ \ x\to -\infty,\end{array}\right.\eeq
moreover,
\beq\label{decay}
\int_0^{+\infty} \E^{C_0 x}(|q(x)| + |q(-x)+c^2|dx<\infty,\ \ \ C_0>c>0,
\eeq
\beq\label{decay1}
\int_{\R}(x^6+1)|q^{(i)}(x)|dx<\infty, \quad i=1,...,11.
\eeq
It is known (cf.\ \cite{EGT}, \cite{ET}), that this Cauchy problem has a unique solution satisfying $q(\cdot,t)\in C^3(\R)$ and
\beq\label{sol}
\int_0^{+\infty}|x|(|q(x,t)| + |q(-x,t)+c^2|)dx<\infty, \qquad t\in\R.
\eeq
In fact, by \cite{Ryb} it will be even real analytic, but we will not use this fact.
From several results (\cite{Bik1}--\cite{Bik2}, \cite{gp1}, \cite{gp2}, \cite{FW},  \cite{N}), obtained on a physical level of rigor, it is
known that the asymptotic behavior of $q(x,t)$ as $t\to\infty$ can be split into three main regions:
\begin{itemize}
\item In the region $x<-6 c^2 t$ the solution is asymptotically close to the background $-c^2$ up to a decaying dispersive tail.
\item In the region $-6c^2 t<x<4c^2 t$ the solution can asymptotically be described by an elliptic wave.
\item In the region $4c^2 t < x$ the solution is asymptotically given by a sum of solitons.
\end{itemize}
This is illustrated in Figure~\ref{fig:num}.
\begin{figure}
\includegraphics[width=8cm]{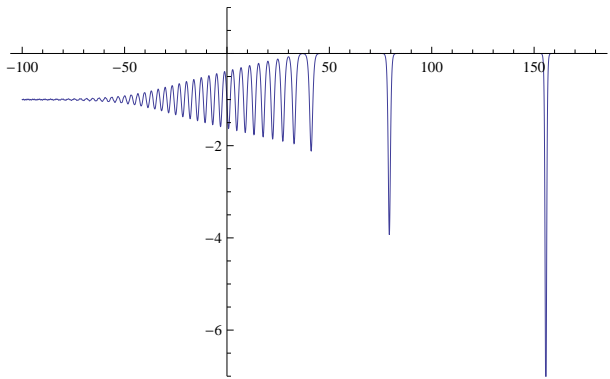}
\caption{Numerically computed solution $q(x,t)$ of the KdV equation at time $t=10$, with initial
condition $q(x,0)=\frac{1}{2}(\mathrm{erf}(x)-1)-5\sech(x-1)$.} \label{fig:num}
\end{figure}
In fact, the long-time asymptotics for this problem were first studied by Gurevich and Pitaevskii \cite{gp1}, \cite{gp2}. These authors have used the Whitham multi-phase averaging method and
obtained the main term of the asymptotics of the solution in terms the Jacobi elliptic function. Moreover, they
gave a qualitative picture of the splitting of an initial step into solitons. Since the Schr\"odinger operator with the Heaviside step function as potential has no discrete spectrum,
this picture refuted the general idea that solitons arise only from the discrete spectrum. This phenomenon was explained by Khruslov \cite{Kh1}, \cite{Kh2} with the
help of the inverse scattering transform (IST) in the form of the Marchenko equation. The IST not only made it possible to obtain an explicit form of these asymptotic solitons but also to give a rigorous proof that the solitons are generated by a small vicinity of the edge of the continuous spectrum.
Further developments of this method can be found in \cite{Ven} and \cite{KKt}.
The first finite-gap description of the asymptotics for the steplike initial problem of the KdV equation was given by Bikbaev and Novokshenov \cite{BikN1} only in 1987 (see also \cite{Bik1}--\cite{Bik3}, \cite{Bik2}, \cite{BikN3}, and the review \cite{N}).
The results are based on an analysis of the Whitham equations and the theory of analytic functions on a hyperelliptic surface.
Our aim here is to use the nonlinear steepest decent method for oscillatory Riemann--Hilbert problem (see \cite{GT} for an introduction to this method in the
case $c=0$ as well as for further references) and apply it to rigorously establish the above mentioned asymptotics.
Related results for an expansive step ($-c^2>0$) can be found in \cite{LN}.

The paper is organized as follows: Section~\ref{sec:rhp} provides some necessary information about the inverse scattering transform on steplike backgrounds. Then
we establish the asymptotics in the soliton region $4c^2 t < x$ in Section~\ref{sec:sr}. In Section~\ref{sec:er}  the initial RH problem is reduced to a "model"
problem in the domain $-6c^2 t<x<4c^2 t$, and in Section~\ref{sec:mp} we solve this model problem. Section~\ref{sec:dr} contains the solution of the model
problem in the domain $x<-6c^2 t$.

Finally, we should remark that our results do not cover the two transitional regions: $4c^2 t \approx x$ near the leading wave front,
and $x\approx-6c^2 t$ near the back wave front. It can be observed numerically that in the first transitional region the modulated elliptic wave develops into
a train of asymptotic solitons.  As pointed out before, these asymptotic solitons have already been
rigorously studied by Khruslov \cite{Kh2}. As for the second region, the matching of the leading asymptotics behind the back front
and in the elliptic region for the modified KdV equation is briefly discussed in \cite[Rem.~4.3]{KM}. However, since the error bounds
obtained from the RHP method break down near the edges, a rigorous justification is beyond the scope of the present paper.

\section{Statement of the RH problem and the first conjugation step}
\label{sec:rhp}

Let $q(x,t)$ be the solution of the Cauchy problem \eqref{kdv}--\eqref{decay1}. Associated with $q(x,t)$ is a self-adjoint Schr{\"o}dinger operator
\begin{equation} \label{defjac}
H(t) = -\frac{d^2}{dx^2}+q(.,t), \qquad \db(H)=H^2(\R) \subset L^2(\R).
\end{equation}
Here $L^2(\R)$ denotes the Hilbert space of square integrable (complex-valued) functions
over $\R$ and $H^k(\R)$ the corresponding Sobolev spaces.

The spectrum of $H$ consists of an absolutely
continuous part $[-c^2,\infty)$ plus a finite number of eigenvalues  $-\kappa_j^2\in(-\infty,-c^2)$,
$1\le j \le N$, where $c<\kappa_1<...<\kappa_N$. In turn, the absolutely continuous part of the spectrum consists of the part $[0,\infty)$ of multiplicity two and the part $[-c^2, 0]$ of multiplicity one. In addition, there exist two Jost solutions $\phi(k,x,t)$ and $\phi_1(k,x,t)$
which solve the differential equation
\beq \label{shturm}
H(t) \psi(k,x,t) = k^2 \psi(k,x,t), \qquad \im (k)> 0,
\eeq
and asymptotically look like the free solutions of the background equations
\beq\label{lims}
\lim_{x \to  +\infty} \E^{- \I kx} \phi(k,x,t) =1,\quad \lim_{x \to  -\infty} \E^{ \I k_1 x} \phi_1(k,x,t) =1.
\eeq
Here $k_1=\sqrt{k^2 +c^2}$, and $k_1>0$ for $k\in[0,\I c)_r$. The last notation means the right side of the cut along the interval $[0,\I c]$. Accordingly, $k_1<0$ for $k\in[0,\I c)_l$, i.e. from the left.
As a function of $k$ the function $\phi(k,x,t)$ (resp., $\phi_1$) is analytic in the domain $\C^U=\{k:\ \im (k) > 0\}$ (resp. $\C^U_c:=\C^U\setminus (0,\I c]$) and continuous
up to the boundary of this domain. Here subscript $U$ corresponds to the upper half plane.

The Jost solutions admit the usual representation via the transformation operators
\beq\label{phipl}
\aligned
\phi(k,x,t)=\E^{\I k x} + & \int_x^{+\infty}K(x,y,t)\E^{\I k y}dy,\\
\phi_1(k,x,t)=\E^{-\I k_1 x} + & \int^x_{-\infty}K_1(x,y,t)\E^{-\I k_1 y}dy,
\endaligned
\eeq
where $K(x,y,t)$ and $K_1(x,y,t)$ are real valued functions, and
\beq\label{1}
K(x,x,t)=\frac 1 2 \int_x^{+\infty}
q(y,t)dy,\quad K_1(x,x,t)=\frac 1 2 \int^x_{-\infty}
(q(y,t)+c^2) dy.
\eeq

Furthermore, one has the scattering relations
\beq \label{rscat}\aligned
T(k,t) \phi_1(k,x,t) = & \ol{\phi(k,x,t)} +
R(k,t) \phi(k,x,t),  \qquad k \in \R,\\
T_1(k,t) \phi(k,x,t) = & \ol{\phi_1(k,x,t)} +
R_1(k,t) \phi_1(k,x,t),  \qquad k_1\in \R,
\endaligned
\eeq
where $T(k,t)$, $R(k,t)$ (resp. $T_1(k,t)$, $R_1(k,t)$) are the right (resp. left) transmission and reflection coefficients. They constitute the entries of the scattering matrix. Denote by
\beq\label{wronsk}
W(k,t)= \phi_1(k,x,t)\phi^\prime(k,x,t) -\phi_1^\prime(k,x,t)\phi(k,x,t)
\eeq
the Wronskian of the Jost solutions, where $f^\prime=\frac{\pa}{\pa x} f$.
In what follows we assume that the initial data \eqref{ini} belong to the generic class of nonresonant potentials for which
\beq\label{nonres}
W(\I c,0)\neq 0.
\eeq
\begin{lemma}[\cite{BF}, \cite{EGT}]\label{lemsc}
The entries of the scattering matrix have the following  properties\footnote{We list here only those of the properties, that are relevant for the present paper}:
\begin{enumerate}[{\bf 1.}]
\item The transmission coefficients $T(k,t)$, $T_1(k,t)$ are meromorphic in the domain $\C^U_c:=\C^U\setminus (0,\I c]$, continuous up to the boundary and
have simple poles at $\I \kappa _1 , \dots, \I \kappa_N$.
The residues of $T(k,t)$ are given by
\beq\label{eq:resT}
\res_{\I \kappa_j} T(k,t) = \I  \mu _j(t) \gamma _{j}(t)^2,
\ \mbox{where}\
\ga_{j}(t)^{-1} = \lVert \phi(\I \kappa_j,.,t)\rVert_2,
\eeq
and $\phi(\I \kappa_j,x,t) = \mu_j(t) \phi_1(\I \kappa_j,x,t)$.
\item Everywhere in the domain $\C^U_c$
\beq\label{ident}
T(k,t)=2\I k W^{-1}(k,t),\ \ T_1(k,t)=2\I k_1 W^{-1}(k,t).
\eeq
\item The reflection coefficient has the symmetry property $R(-k,t)=\ol{R(k,t)}$ as $k\in\R$ and
\beq\label{realos}
\ol {T_1(k,t)} T(k,t)=1-|R(k,t)|^2,\ \ol {R(k,t)}T(k,t) + R_1(k,t)\ol {T(k,t)}=0,\ \ k\in\R;
\eeq
\item The functions $R_1, T_1$ and $T$ also possess the symmetry property with respect to $k_1\in\R$, in particular, $T_1(k(-k_1),t)=\ol{T_1(k(k_1),t)}$.
Moreover,
\beq\label{posl}
-T(k,t)\ol{T^{-1}(k,t)}=T_1(k,t)\ol{T_1^{-1}(k,t)}
=R_1(k,t),\quad k_1\in[-c,c].\eeq
\item
The time evolutions of the quantities $\ga_j(t)$, $R(k,t)$ and $|T(k,t)|^2$  are given by
$R(k,t)= R(k) \E^{8 \I k^3 t}$ for $ k\in\R$, $|T(k,t)|^2=|T(k)|^2 \E^{8 \I k^3 t}$ for $ k\in[-\I c, \I c]$, and $\ga_j(t) = \ga_j \E^{4 \kappa_j^3 t},$
where $\ga_j=\ga_j(0)$, $R(k)=R(k,0)$ and $T(k)=T(k,0)$.
\item  Under the assumption \eqref{decay} the function $R(k)$ admits an analytic continuation to the domain $\{k:\ 0<\im k< C_0\}\setminus (0, \I c]$
preserving the symmetry property  $R(k(-k_1))=\ol{R(k(k_1))}$ for $k_1\in(-c,c)$.
\end{enumerate}
\end{lemma}
The properties, cited in this lemma, belong to the list of necessary and sufficient properties of the scattering data for the
step-like potential with prescribed behavior of the perturbations. All of them, except of the last one, are valid  for much wider
class of perturbations than the class \eqref{decay}, for example, for the class of potentials with a finite first moment of perturbations.

Consider  a vector-function $m(k,x,t)$ as a function of spectral parameter $k$, $k\in \C\setminus (\R \cup[-\I c, \I c])$,
where $x,t$ are fixed parameters. We define this vector-function as follows
\beq\label{defm}
m(k,x,t)= \left\{\begin{array}{c@{\quad}l}
\begin{pmatrix} T(k,t) \phi_1(k,x,t) \E^{\I kx},  & \phi(k,x,t) \E^{-\I kx} \end{pmatrix},
& k\in \C^U_c, \\
\begin{pmatrix} \phi(-k,x,t) \E^{\I kx}, & T(-k,t) \phi_1(-k,x,t) \E^{-\I kx} \end{pmatrix},
& k\in\C^L_c,
\end{array}\right.
\eeq
where $\C^U_c:=\{k:\ \im k>0\}\setminus(0,\I c]$,
$\C^L_c:=\{k:\ \im k<0\}\setminus(0,-\I c]$.
\begin{lemma}\label{asypm}
The function $m(k)=m(k,x,t)$, defined by formula \eqref{defm}, has the following asymptotical behavior
\beq\label{asm}
m(k,x,t)= (1,1) -\frac{1}{2\I k}\left(\int_x^{+\infty}q(y,t)dy\right) (-1,1) + O\left(\frac{1}{k^2}\right).
\eeq
\end{lemma}

\begin{proof}
Will be given in Appendix~\ref{sec:app}.
\end{proof}

Next we introduce
\beq\label{tau}
 \chi(k):=-\lim_{\varepsilon\to +0} \ol{T(k+\varepsilon,0)}T_1(k+\varepsilon,0),
 \qquad \mbox{for}\quad k\in[0,\I c],
\eeq
and continue this function on the interval $[-\I c,0]$ by
\beq\label{proptau}
\chi(-k)=-\chi(k),\qquad k\in[-\I c,0].
\eeq
Equation \eqref{ident}  then implies
\beq\label{proptau1}
\frac{\chi(k)}{k}>0\quad \mbox{for}\quad k\in [-\I c, \I c].
\eeq
We are interested in the jump condition of $m(k,x,t)$ on the contours $\Sigma\cup\Sigma_c$, where $\Sigma=\R$, oriented
LTR (left-to-right), and $\Sigma_c=[\I c, -\I c]$, oriented top-down.

In general, for an oriented contour $\Sigma$, the value $m_+(k)$ (resp.\
$m_-(k)$) will denote the nontangential limit of $m(\kappa)$ as
$\kappa\to k$ from the positive (resp.\ negative) side of $\Sigma$.
Here the positive (resp.\ negative) side is the one which lies to
the left (resp.\ right) as one traverses the contour in the
direction of its orientation.  In order to not mix up  limit
values of functions from the different sides of contours with
another meaning of signs $+$ and $-$, in what follows we denote the
upper (resp. lower) half plane as $\C^U$ (resp. $\C^L$). Any
notation, which is connected with upper or lower half plane, will be
also marked by subscript $U$ or $L$. For example, $\Sigma_c^U=[\I c,
0]$. Moreover, by  subscripts $l$ and $r$ we will mark, when
necessary, the values of functions from the left and right of the
cut $[-\I c, \I c]$. In particular, as the definition (see
\eqref{tau}) of the function $\chi$  we could write $\chi=-[\ol T
T_1]_+$ or $\chi=-[\ol T T_1]_r$. Note also, that the reflection
coefficient $R(k)$, $k\in \R$, the function $\chi(k)$, $k\in[-\I c, \I
c]$ and the discrete spectrum together with right normalizing
constants $(\kappa_j, \ga_j),$ $1\le j \le N$, completely define
the kernel of the right Marchenko equation, and, therefore, the
potential $q(x)$ (cf.\cite{BF}, \cite{EGT}). That is why we refer to
them as the minimal scattering data of operator $H(0)$.

\begin{theorem}\label{thm:vecrhp}
Let $\{ R(k),\; k\in \R; \chi(k), \ k\in[-\I c, \I c];\  (\kappa_j, \ga_j), \: 1\le j \le N \}$ be
the minimal scattering data of the operator $H(0)$. Then $m(k)=m(k,x,t)$ defined in \eqref{defm}
is a solution of the following vector Riemann--Hilbert problem.

Find a vector-valued function $m(k)$ which is meromorphic away from $\Sigma\cup\Sigma_c$ with simple poles at
$\pm\I\kappa_j$ and satisfies:
\begin{enumerate}
\item The jump condition $m_+(k)=m_-(k) v(k)$
\beq \label{eq:jumpcond}
v(k)=\left\{\begin{array}{cc}\begin{pmatrix}
1-|R(k)|^2 & - \ol{R(k)} \E^{-t\Phi(k)} \\
R(k) \E^{t\Phi(k)} & 1
\end{pmatrix},& k\in\Sigma,\\
 \ &\ \\
\begin{pmatrix}
1 & 0 \\
\chi(k) \E^{t\Phi(k)} & 1
\end{pmatrix},& k\in\Sigma_c^U,\\
 \ &\ \\
\begin{pmatrix}
1 & \chi(k) \E^{-t\Phi(k)} \\
0 & 1
\end{pmatrix},& k\in\Sigma_c^L,\\
\end{array}\right.
\eeq

\item
the pole conditions
\beq\label{eq:polecond}
\aligned
\res_{\I\kappa_j} m(k) &= \lim_{k\to\I\kappa_j} m(k)
\begin{pmatrix} 0 & 0\\ \I \ga_j^2 \E^{t\Phi(\I \kappa_j)}  & 0 \end{pmatrix},\\
\res_{-\I\kappa_j} m(k) &= \lim_{k\to -\I\kappa_j} m(k)
\begin{pmatrix} 0 & - \I \ga_j^2 \E^{t\Phi(\I \kappa_j)} \\ 0 & 0 \end{pmatrix},
\endaligned
\eeq
\item
the symmetry condition
\beq \label{eq:symcond}
m(-k) = m(k) \sigI,
\eeq
\item
the normalization condition
\beq\label{eq:normcond}
\lim_{\kappa\to\infty} m(\I\kappa) = (1\quad 1).
\eeq
\end{enumerate}
Here the phase $\Phi(k)=\Phi(k,x,t)$ is given by
\begin{equation}
\Phi(k)= 8 \I k^3+2\I k \frac {x}{t},
\end{equation}
\end{theorem}

\begin{proof}
Will be given in Appendix~\ref{sec:app}.
\end{proof}

We note that $m(z)$ defined in \eqref{defm} is the only solution of the above Riemann--Hilbert problem.
This can by seen after rewriting the pole conditions as jump conditions (see below) from \cite[Thm.~3.2]{GT}
(or alternatively from \cite[Thm.~4.3]{MT}). Since all conjugation and deformation steps applied below
are reversible, the solutions of all further Riemann--Hilbert problems will be unique as well. In this respect note that
our jump matrix satisfies
\beq\label{proverka}
v(-k) = \sigI v(k)^{-1} \sigI,\quad k\in\hat\Sigma;
\eeq
and $\det(v(k))=1$.

For our further analysis we rewrite the pole condition as a jump
condition and hence turn our meromorphic Riemann--Hilbert problem into a holomorphic Riemann--Hilbert problem following literally \cite{GT}.
Choose $\eps>0$ so small that the discs $\left\vert k- \I \kappa_j \right\vert<\eps$ lie inside the the domain $\C^U_c$ and
do not intersect any of the other contours.
Denote the circle boundaries of these small discs as $\T_j^U$. Redefine $m(k)$ in a neighborhood of $\I \kappa_j$ respectively $- \I \kappa_j$ according to
\beq\label{eq:redefm}
m(k) = \begin{cases} m(k) \begin{pmatrix} 1 & 0 \\
-\frac{\I \gamma_j^2 \E^{t\Phi(\I \kappa_j)} }{k- \I \kappa_j} & 1 \end{pmatrix},  &
|k- \I \kappa_j|< \eps,\\
m(k) \begin{pmatrix} 1 & \frac{\I \gamma_j^2 \E^{t\Phi(\I \kappa_j)} }{k+ \I \kappa_j} \\
0 & 1 \end{pmatrix},  &
|k+ \I \kappa_j|< \eps,\\
m(k), & \text{else}.\end{cases}
\eeq
 Note that in $\C^L_c$ we redefined $m(k)$ such that it respects our symmetry \eqref{eq:symcond}. Then a straightforward calculation using
$\res_{\I \kappa} m(k) = \lim_{k\to\I\kappa} (k-\I \kappa)m(k)$ shows the following well-known result:

\begin{lemma}[\cite{GT}]\label{lem:holrhp}
Suppose $m(k)$ is redefined as in \eqref{eq:redefm}. Then $m(k)$ is holomorphic in $\C\setminus\left(\Sigma\cup \Sigma_c\cup\cup_{j=1}^N (\T_j^U\cup \T_j^L)\right)$. Furthermore it satisfies \eqref{eq:jumpcond}, \eqref{eq:symcond}, \eqref{eq:normcond}
and
\beq \label{eq:jumpcond2}
\aligned
m_+(k) &= m_-(k) \begin{pmatrix} 1 & 0 \\
-\frac{\I \gamma_j^2 \E^{t\Phi(\I\kappa_j)}}{k-\I \kappa_j} & 1 \end{pmatrix},\quad k\in \T_j^U,\\
m_+(k) &= m_-(k) \begin{pmatrix} 1 & -\frac{\I \gamma_j^2 \E^{t\Phi(\I \kappa_j)}}{k+ \I \kappa_j} \\
0 & 1 \end{pmatrix},\quad k\in\T_j^L,
\endaligned
\eeq
where the small circle around $\I \kappa_j$ is oriented counterclockwise and the one around $-\I \kappa_j$ is oriented clockwise.
\end{lemma}

\section{Asymptotics in the domain $4c^2t<x$}
\label{sec:sr}

To reduce our RH problem to a model problem, that can be solved explicitly, we will use the well-known conjugation and
deformation techniques.

\begin{lemma}[Conjugation]\label{lem:conjug}
Let $m$ be the solution of the RH problem $m_+(k)=m_-(k) v(k)$, $k\in\Sigma$.
Assume that $\wti{\Sigma}\subseteq\Sigma$. Let $D$ be a matrix of the form
\beq\label{defmd}
D(k) = \begin{pmatrix} d(k)^{-1} & 0 \\ 0 & d(k) \end{pmatrix},
\eeq
where $d: \C\backslash\wti{\Sigma}\to\C$ is a sectionally analytic function. Set
\beq
\ti{m}(k) = m(k) D(k),
\eeq
then the jump matrix transforms according to
\beq
\ti{v}(k) = D_-(k)^{-1} v(k) D_+(k).
\eeq
If $d$ satisfies $d(k)\neq 0$, $d(-k) = d(k)^{-1}$ for $k\in\C\setminus\wti{\Sigma}$ and $\lim_{\kappa\to\infty} d(\I\kappa)=1$, then the transformation $\ti{m}(k) = m(k) D(k)$
respects the symmetry and normalization conditions \eqref{eq:symcond} and \eqref{eq:normcond}, respectively.
\end{lemma}

In particular, we obtain
\beq
\ti{v} = \begin{pmatrix} v_{11} & v_{12} d^{2} \\ v_{21} d^{-2}  & v_{22} \end{pmatrix},
\qquad k\in\hat\Sigma\backslash\wti{\Sigma},
\eeq
respectively
\beq
\ti{v} = \begin{pmatrix} \frac{d_-}{d_+} v_{11} & v_{12} d_+ d_- \\
v_{21} d_+^{-1} d_-^{-1}  & \frac{d_+}{d_-} v_{22} \end{pmatrix},
\qquad k\in\wti{\Sigma}.
\eeq
Now we make the first conjugation step, which allows us to take into account the influence of the discrete spectrum.
To this end we will need the value $\kappa_0$ defined via $\re(\Phi(\I\kappa_0))=0$, that is,
\[
 \kappa_0 = \sqrt{\frac{x}{4 t}}>0.
\]
We will set $\kappa_0=0$ if $\frac{x}{t}<0$ for notational convenience. Then we have $\re(\Phi(\I \kappa_j))>0$
for all $\kappa_j > \kappa_0$ and  $\re(\Phi(\I \kappa_j))<0$ for all $\kappa_j < \kappa_0$. Hence, in the first case
the off-diagonal entries of our jump matrices are exponentially growing and we need to turn them into exponentially
decaying ones. Therefore we set
\[
\Lambda(k):=\prod_{\kappa_j > \kappa_0} \frac{k+\I\kappa_j}{k-\I\kappa_j},
\]
and  introduce the matrix
\beq
D(k) = \left\{\begin{array}{lll}
\begin{pmatrix} 1 & -\frac{k-\I\kappa_j}{\I\gamma_j^2 \E^{t\Phi (\I\kappa_j)}}\\
\frac{\I\gamma_j^2 \E^{t\Phi(\I\kappa_j)}}{k-\I\kappa_j} & 0 \end{pmatrix}
D_0(k), &  |k-\I\kappa_j|<\eps, & j=1,...,N,\\
\begin{pmatrix} 0 & -\frac{\I\gamma_j^2 \E^{t\Phi (\I\kappa_j)}}{k+\I\kappa_j} \\
\frac{k+\I\kappa_j}{\I\gamma_j^2 \E^{t\Phi(\I\kappa_j)}} & 1 \end{pmatrix}
D_0(k), & |k+\I\kappa_j|<\eps, & j=1,...,N,\\
\ & \ & \ \\
D_0(k), & \text{else},&
\end{array}\right.
\eeq
where
\[
D_0(k) = \begin{pmatrix} \Lambda(k)^{-1} & 0 \\ 0 & \Lambda(k) \end{pmatrix}.
\]
Observe that by $\Lambda(-k)=\Lambda^{-1}(k)$ we have
\beq\label{defP}
D(-k)= \sigI D(k) \sigI.
\eeq
Now we set
\beq\label{def:mti}
\ti{m}(k)=m(k) D(k).
\eeq
Note that by \eqref{defP} this conjugation preserve properties \eqref{eq:symcond} and \eqref{eq:normcond}.

Then (for details see Lemma~4.2 of \cite{GT}) the jump
corresponding to $ \kappa_0<\kappa_j$ is given by
\beq\label{jumpcondti}
\aligned
\ti{v}(k) &= \begin{pmatrix}1& -\frac{(k-\I\kappa_j)\Lambda^2(k)}
{\I\gamma_j^2 \E^{t\Phi (\I\kappa_j)}}\\ 0 &1\end{pmatrix},
\qquad k\in \T_j^U, \\
\ti{v}(k) &= \begin{pmatrix}1& 0 \\ -\frac{k+\I\kappa_j}
{\I\gamma_j^2 \E^{t\Phi(\I\kappa_j)} \Lambda^2(k)}&1\end{pmatrix},
\qquad k\in\T_j^L,
\endaligned
\eeq
and the jumps corresponding to $\kappa_0>\kappa_j$ (if any) by
\beq
\aligned
\ti{v}(k) &= \begin{pmatrix} 1 & 0 \\ -\frac{\I\gamma_j^2 \E^{t\Phi(\I\kappa_j)} \Lambda(k)^{-2}}{k-\I\kappa_j}
 & 1 \end{pmatrix},
\qquad  k\in \T_j^U, \\
\ti{v}(k) &= \begin{pmatrix} 1 & -\frac{\I\gamma_j^2 \E^{t\Phi(\I\kappa_j)} \Lambda(k)^2}{k+\I\kappa_j} \\
0 & 1 \end{pmatrix},
\qquad  k\in \T_j^L.
\endaligned
\eeq
In particular, all jumps corresponding to poles, except for possibly one if
$\kappa_j=\kappa_0$, are exponentially close to the identity for $t\to\infty$. In the latter case we will keep the
pole condition for $\kappa_j=\kappa_0$ which now reads
\beq
\aligned
\res_{\I\kappa_j} \ti{m}(k) &= \lim_{k\to\I\kappa_j} \ti{m}(k)
\begin{pmatrix} 0 & 0\\ \I\gamma_j^2 \E^{t\Phi(\I\kappa_j)} \Lambda(\I\kappa_j)^{-2}  & 0 \end{pmatrix},\\
\res_{-\I\kappa_j} \ti{m}(k) &= \lim_{k\to -\I\kappa_j} \ti{m}(k)
\begin{pmatrix} 0 & -\I\gamma_j^2 \E^{t\Phi(\I\kappa_j)} \Lambda(\I\kappa_j)^{-2} \\ 0 & 0 \end{pmatrix}.
\endaligned
\eeq
Furthermore, the jump along $\Sigma\cup\Sigma_c$ now reads
\beq \label{jumpcond3}
\ti v(k)=\left\{\begin{array}{cc}\begin{pmatrix}
1-|R(k)|^2 & - \Lambda^2(k)\ol{R(k)} \E^{-t\Phi(k)} \\
\Lambda^{-2}(k)R(k) \E^{t\Phi(k)} & 1
\end{pmatrix},& k\in\Sigma,\\
 \ &\ \\
\begin{pmatrix}
1 & 0 \\
\Lambda^{-2}(k)\chi(k) \E^{t\Phi(k)} & 1
\end{pmatrix},& k\in\Sigma_c^U,\\
 \ &\ \\
\begin{pmatrix}
1 & \Lambda^{2}(k)\chi(k) \E^{-t\Phi(k)} \\
0 & 1
\end{pmatrix},& k\in\Sigma_c^L.\\
\end{array}\right.
\eeq
The new Riemann--Hilbert problem \beq\label{tim}\ti m_+(k)=\ti m_-(k)\ti v(k)\eeq for the vector $\ti m$ preserves its asymptotics \eqref{eq:normcond}
as well as the symmetry condition \eqref{eq:symcond}. In particular, after conjugation all jumps corresponding to poles are now exponentially close to the identity as $t\to\infty$.
To turn the remaining jumps along $\Sigma\cup\Sigma_c$ into this form as well, we chose two  contours $\Sigma^U$ and $\Sigma^L$, which are symmetric with respect to map $k\mapsto -k$, enclose $\Sigma_c$ and do not enclose points of discrete spectrum between them, and are sufficiently close
to the original contour $\Sigma=\mathbb R$, such that $\Sigma^U\cup\Sigma^L\subset \{k:\, |\im k|<C_0\}.$ Lemma \ref{lemsc}, item {\bf 6}, guarantees that function $R(k)$ is analytic in the region $\Omega^U$ between $\Sigma^U$ and $\mathbb R\cup \Sigma^U_c$ (cf.\ Figure~\ref{fig:sol}). Continue function $R(k)$ in the domain $ \{k:\, -C_0<\im k <0\}\setminus(-\I c, 0]$ by formula
\beq\label{contR}
R(k)=\ol{R(-k)}.
\eeq Then the function $\ol R$ is analytic in the domain $\Omega^L$ (cf.\ Figure~\ref{fig:sol}).
\begin{figure}[h]
\begin{picture}(8,4)
\put(4,1){\line(0,1){2}}
\put(4,1.8){\vector(0,-1){0.1}}

\put(4.2,2.9){$\I c$}
\put(4.2,1){$-\I c$}
\put(4,1){\circle*{0.1}}
\put(4,3){\circle*{0.1}}

\put(4.2,3.9){$\I \kappa_1$}
\put(4.2,0){$-\I \kappa_1$}
\put(4,0.2){\circle*{0.1}}
\put(4,3.8){\circle*{0.1}}
\put(4,0.2){\circle{0.3}}
\put(4,3.8){\circle{0.3}}

\put(3.1,3.4){$\Sigma^U$}
\put(3.1,0.3){$\Sigma^L$}

\put(3.3,2.25){$\Omega^U$}
\put(3.3,1.5){$\Omega^L$}

\put(4.2,1.9){$\Sigma_c$}

\curve(0.533, 2.303,1.067, 2.316,1.6, 2.367,2.133, 2.51,2.667,
2.793,3.2, 3.171,3.733, 3.458,4.267, 3.458,4.8, 3.171,5.333,
2.793,5.867, 2.51,6.4, 2.367,6.933, 2.316,7.467, 2.303)
\curve(0.533, 1.697,1.067, 1.684,1.6, 1.633,2.133, 1.49,2.667,
1.207,3.2, 0.829,3.733, 0.542,4.267, 0.542,4.8, 0.829,5.333,
1.207,5.867, 1.49,6.4, 1.633,6.933, 1.684,7.467, 1.697)

\put(6.9,2.31){\vector(1,0){0.1}}
\put(6.9,1.68){\vector(1,0){0.1}}
\put(1.1,2.31){\vector(1,0){0.1}}
\put(1.1,1.68){\vector(1,0){0.1}}

\curvedashes{0.05,0.05}
\curve(0.3,2, 7.7,2)

\end{picture}
\caption{Contour deformation in the soliton region.}\label{fig:sol}
\end{figure}
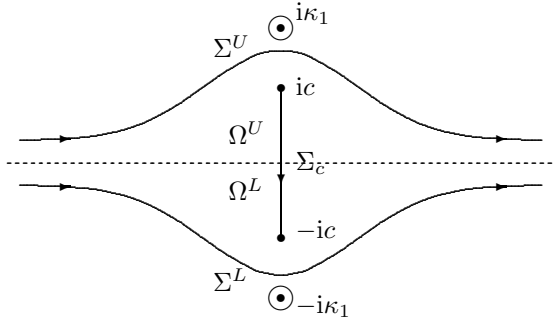

Now we factorize the jump matrix along $\Sigma$ according to
\beq
\hat v=b_L^{-1}b_U=\begin{pmatrix} 1 & - \Lambda^2(k)\ol{R(k)} \\
0&1\end{pmatrix}\begin{pmatrix} 1&0\\
\Lambda^{-2}(k)R(k) & 1
\end{pmatrix}
\eeq
and set
\beq
\hat m(k) =\left\{\begin{array}{ll}\hat m(k) b_U^{-1}(k), & k\in\Omega^U,\\
\ti m(k) b_L^{-1}(k), & k\in\Omega^L,\\
\ti m(k), & \text{else},
\end{array}\right.
\eeq
such that the jump along $\Sigma$ is moved to $\Sigma^U\cup\Sigma^L$ and given by
\beq
\hat v(k) =\left\{\begin{array}{ll}
\begin{pmatrix} 1&0\\
\Lambda^{-2}(k)R(k) & 1
\end{pmatrix}, & k\in\Sigma^U,\\
\ & \ \\
\begin{pmatrix} 1 & - \Lambda^2(k)\ol{R(k)} \\
0&1\end{pmatrix}, & k\in\Sigma^L.
\end{array}\right.
\eeq
The jumps along the circles $\T_j^U\cup\T_j^L$ are unchanged and the jump along $\Sigma_c$ now reads
\beq
\hat v(k)=\left\{\begin{array}{ll}
\begin{pmatrix}
1 & 0 \\
(R_- - R_+ +\chi)\Lambda^{-2} & 1
\end{pmatrix},& k\in\Sigma_c^U,\\
\ & \ \\
\begin{pmatrix}
1 & (\ol R_- -\ol R_++\chi)\Lambda^{2}  \\
0 & 1
\end{pmatrix},& k\in\Sigma_c^L.
\end{array}\right.
\eeq
The following lemma shows that this jump in fact also disappears.

\begin{lemma}\label{lemzero}
The following identities are valid:
\[
R_-(k) - R_+(k) +\chi(k)=0,\qquad k\in\Sigma^U_c,
\]
\[
\ol R_-(k) - \ol R_+(k) +\chi(k)=0,\qquad k\in\Sigma^L_c.
\]
\end{lemma}

\begin{proof}
With the help of the Pl\"ucker identity (cf.\ \cite{Te09}) and by use of \eqref{contR} and \eqref{proptau}.
\end{proof}

Hence, all jumps $\hat{v}$ are exponentially close to the identity as $t\to\infty$ and one
can use Theorem~A.6 from \cite{KTa} to obtain (repeating literally the proof of Theorem~4.4 in \cite{GT})
the following result:

\begin{theorem}\label{thm:asym}
Assume \eqref{decay}--\eqref{decay1} and abbreviate by $c_j= 4 \kappa_j^2$
the velocity of the $j$'th soliton determined by $\re(\Phi(\I \kappa_j))=0$.
Then the asymptotics in the soliton region, $x/t - 4c^2 \geq \epsilon$ for some small
$\epsilon>0$, are as follows:

Let $\eps > 0$ be sufficiently small such that the intervals
$[c_j-\eps,c_j+\eps]$, $1\le j \le N$, are disjoint and lie inside $(4c^2,\infty)$.

If $|\frac{x}{t} - c_j|<\eps$ for some $j$, one has
\begin{align}
q(x,t)& = \frac{-4\kappa_j\gamma_j^2(x,t)}{(1+(2\kappa_j)^{-1}\gamma_j^2(x,t))^2} +O(t^{-l})
\end{align}
for any $l\in\N$, where
\beq
\gamma_j^2(x,t) = \gamma_j^2 \E^{-2\kappa_j x + 8 \kappa_j^3 t} \prod_{i=j+1}^N \left(\frac{\kappa_i-\kappa_j}{\kappa_i+\kappa_j}\right)^2.
\eeq

If $|\frac{x}{t} -c_j| \geq \eps$, for all $j$, one has
\begin{align}
q(x,t)& = O(t^{-l})
\end{align}
for any $l\in\N$.
\end{theorem}

\section{Reduction to the model problem in the domain $-6c^2t<x<4c^2t$}
\label{sec:er}

Now we turn to the elliptic region $-6c^2t<x<4c^2t$ we first proceed as in the previous section to obtain $\ti v$ where we now use
\beq\label{newlambda}
\Lambda(k):=\prod_{j=1}^N \frac{k+\I\kappa_j}{k-\I\kappa_j}
\eeq
since clearly $\kappa_j > c > \kappa_0$ for all $j$. For the conjugation step we will use a $g$-function as first outlined in \cite{dvz}.
Our approach here is similar to \cite{aik} and \cite{KM}.

Set $\xi=\frac{x}{12t}$, then $\Phi(k)=\Phi(k,\xi)=8\I k^3 + 24\I k\xi$. Following \cite{KM} in the domain $\C\setminus\Sigma_c$ introduce the function
\beq\label{deffung}
{g(k):=g(k,x,t)=12\int_{\I c}^k (k^2 + \mu^2)\sqrt\frac{k^2 + a^2}{k^2 + c^2}} dk
\eeq
where the parameters $a=a(\xi)$, $0<a<c$ and $\mu=\mu(\xi)$, $0<\mu<a<c$,
\beq\label{choimu}
\mu^2=\xi +\frac{c^2-a^2}{2}
\eeq
are chosen to satisfy conditions
\beq
\int_0^{\I a}(k^2 + \mu^2)\left[\sqrt{\frac{k^2 + a^2}{k^2 + c^2}}\ \right]_r dk =0
\eeq
and
\beq\label{asympg}
g(k)-4k^3 - 12 k\xi\to 0,\qquad k\to\infty.
\eeq
As is shown in \cite{KM}, these conditions can be satisfied for all values of parameter $\xi$ in the domain $-\frac{c^2}{2}<\xi<\frac{c^2}{3}$.
Set $\Sigma_a=[\I a, -\I a ]$ with the orientation top-down.

\begin{lemma}[\cite{KM}]\label{KM} The function $g(k)$ possess the following properties
\begin{enumerate}[{\bf (a)}]
\item Function $g$ is an odd function in the domain $\C\setminus\Sigma_c$, $g(k)=-g(-k)$;
    \item $g_-(k)+g_+(k)=0$ as $k\in\Sigma_c\setminus\Sigma_a$;
\item $g_-(k) - g_+(k)=B$ as $k\in\Sigma_a$, where $ B:=B(\xi)=2g_+(\I a)>0$;
\item the asymptotical behavior holds as $k\to\infty$: \beq\label{asp}\frac{1}{2}\Phi(k,\xi)-\I g(k,\xi)=\frac{12\xi(c^2 - a(\xi)^2) + 3c^4 + 9a(\xi)^4-6a(\xi)^2c^2}{2k\I} + O\left(\frac{1}{k^3}\right).\eeq
\end{enumerate}
\end{lemma}

\begin{proof}
The last property follows immediately from \eqref{deffung}--\eqref{asympg} and property {\bf (a)}.
\end{proof}

The signature table for the function $\im g(k)$ is depicted in Figure~\ref{fig1}.

\begin{figure}[h]
\begin{picture}(8,4)
\put(0,2){\vector(1,0){8}}
\put(4,0){\line(0,1){1}}
\put(4,3){\line(0,1){1}}

\put(0.5,2.3){$+$}
\put(0.5,1.6){$-$}
\put(7.5,2.3){$+$}
\put(7.5,1.6){$-$}
\put(3.6,3.9){$-$}
\put(3.6,0){$+$}

\put(4.2,2.9){$\I a$}
\put(4.2,1){$-\I a$}
\put(4,1){\circle*{0.1}}
\put(4,3){\circle*{0.1}}

\put(4.2,3.9){$\I c$}
\put(4.2,0){$-\I c$}
\put(4,0){\circle*{0.1}}
\put(4,4){\circle*{0.1}}

\curve(4,3, 5.5,2.5, 7.5,4)
\curve(4,1, 5.5,1.5, 7.5,0)
\curve(4,3, 2.5,2.5, 0.5,4)
\curve(4,1, 2.5,1.5, 0.5,0)

\curvedashes{0.05,0.05}
\curve(4,1, 4,3)

\end{picture}
  \caption{Sign of $\im(g)$}\label{fig1}
\end{figure}
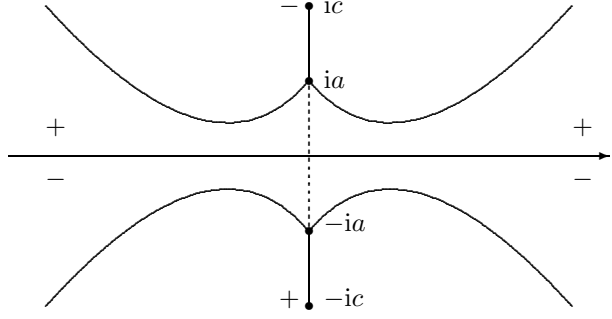

Introduce the function \beq\label{defdg}d(k,t)=\exp(t\Phi(k)/2 - \I t g(k)).\eeq According to \eqref{asp}
we have
\beq\label{asdt}
d(k,t) = 1+ t\frac{z(\xi)}{k\I} + O\left(\frac{1}{k^3}\right),\quad z(\xi)=\frac{12\xi(c^2 - a(\xi)^2) + 3c^4 + 9a(\xi)^4-6a(\xi)^2c^2}{2}.
\eeq
Since the functions $\Phi(k)$ and $g(k)$ are both odd functions of $k$, the function $d(\cdot,t)$ is analytic in $\C\setminus\Sigma_c$ and
satisfies $d(-k,t)=d^{-1}(k,t)$ plus $d(k,t)\to 1$ as $k\to\infty$.
Let $\ti m(k)$ be the solution of the problem \eqref{jumpcondti}--\eqref{tim}.  Set $\hat m(k)=\ti m(k)D(k,t)$, where the diagonal matrix $D(k,t)$ is defined
by \eqref{defmd} with $d(k,t)$, defined by \eqref{defdg}. Applying Lemma~\ref{lem:conjug}
we arrive at the following Riemann--Hilbert problem:
\beq\label{defhatm}
\hat m_+(k)=\hat m_-(k)\hat v(k),\qquad \hat m(k)\to(1,1),\ \  k\to\infty,
\eeq
where
\beq\label{jumpcondh}
\hat{v}(k) = \left\{ \begin{array}{ll} \begin{pmatrix}1& h_j^U(k,\xi,t)\\ 0 &1\end{pmatrix},\qquad  k\in \T_j^U,
&j=1,...,N, \\
\ & \ \\
 \begin{pmatrix}1& 0 \\ h_j^L(k,\xi,t)&1\end{pmatrix},\qquad k\in\T_j^L,
& j=1,...,N,\\
\ & \ \\
 \begin{pmatrix}
1-|R(k)|^2 & - \Lambda^2(k)\ol{R(k)} \E^{-2\I t g(k)} \\
\Lambda^{-2}(k)R(k) \E^{2\I t g(k)} & 1
\end{pmatrix},& k\in\Sigma,\\
\ & \ \\
\begin{pmatrix}
\E^{\I t (g_+ - g_-)} & 0 \\
\Lambda^{-2}(k)\chi(k) \E^{\I t(g_+ + g_-)} & \E^{-\I t (g_+ - g_-)}
\end{pmatrix},& k\in\Sigma_c^U,\\
\ & \ \\
\begin{pmatrix}
\E^{\I t (g_+ - g_-)} & \Lambda^{2}(k)\chi(k) \E^{-\I t(g_+ + g_-)} \\
0 & \E^{-\I t (g_+ - g_-)}
\end{pmatrix},& k\in\Sigma_c^L,
\end{array}\right.
\eeq
where the entries
\beq\label{defhu}
h_j^U(k,\xi,t)=\I\gamma_j^{-2} (k-\I\kappa_j)\Lambda^2(k)
\E^{t(\Phi(k) - \Phi (\I\kappa_j)-2\I g(k))},\eeq
\beq\label{defhu1}
h_j^L(k,\xi,t)=\I\gamma_j^{-2}(k+\I\kappa_j)
\Lambda^{-2}(k)
\E^{-t(\Phi(k) - \Phi (-\I\kappa_j)-2\I g(k))},
\eeq
of the conjugation matrices on the circles decay exponentially with respect to $t$.

Introduce two domains $\Omega^U$ and $\Omega^L$, bounded by $\Sigma$ and contours $\Sigma^U$ and $\Sigma^L$ respectively,
where the contours $\Sigma^U$ and $\Sigma^L$ are symmetric with respect to map $k\mapsto -k$ and oriented LTR (cf.\ Figure~\ref{fig2}).
Moreover, $\Omega^U$ and $\Omega^L$ must remain in the region where $\im(g)>0$ and $\im(g)<0$, respectively.

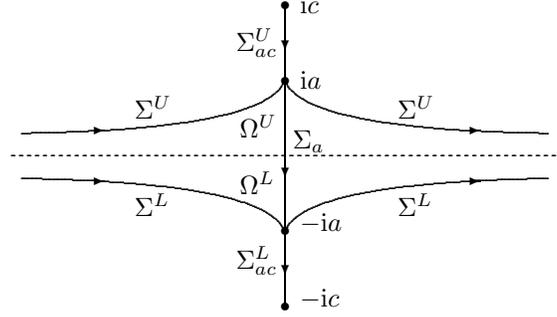
\begin{figure}[h]
\begin{picture}(8,4)
\put(4,0){\line(0,1){4}}
\put(4,3.5){\vector(0,-1){0.1}}
\put(4,1.8){\vector(0,-1){0.1}}
\put(4,0.5){\vector(0,-1){0.1}}

\put(4.2,2.9){$\I a$}
\put(4.2,1){$-\I a$}
\put(4,1){\circle*{0.1}}
\put(4,3){\circle*{0.1}}

\put(4.2,3.9){$\I c$}
\put(4.2,0){$-\I c$}
\put(4,0){\circle*{0.1}}
\put(4,4){\circle*{0.1}}

\put(2,2.5){$\Sigma^U$}
\put(2,1.2){$\Sigma^L$}
\put(5.5,2.5){$\Sigma^U$}
\put(5.5,1.2){$\Sigma^L$}

\put(3.35,3.4){$\Sigma^U_{ac}$}
\put(3.35,0.5){$\Sigma^L_{ac}$}

\put(3.4,2.25){$\Omega^U$}
\put(3.4,1.5){$\Omega^L$}

\put(4.1,2.1){$\Sigma_a$}

\curve(4,3, 5,2.5, 7.5,2.3)
\curve(4,1, 5,1.5, 7.5,1.7)
\curve(4,3, 3,2.5, 0.5,2.3)
\curve(4,1, 3,1.5, 0.5,1.7)

\put(6.5,2.34){\vector(1,0){0.1}}
\put(6.5,1.66){\vector(1,0){0.1}}
\put(1.5,2.34){\vector(1,0){0.1}}
\put(1.5,1.66){\vector(1,0){0.1}}

\curvedashes{0.05,0.05}
\curve(0.3,2, 7.7,2)

\end{picture}
\caption{The first deformation step}\label{fig2}
\end{figure}

Following the standard procedure (see, for example, \cite{dz}, \cite{GT}, \cite{KM}) we factorize the matrix $\hat v(k)$ on the real axis according to
\beq\label{razvod1}
\hat v=b_L^{-1}b_U=\begin{pmatrix} 1 & - \Lambda^2(k)\ol{R(k)} \E^{-2\I t g(k)} \\
0&1\end{pmatrix}\begin{pmatrix} 1&0\\
\Lambda^{-2}(k)R(k) \E^{2\I t g(k)} & 1
\end{pmatrix}.
\eeq
Set
\beq\label{razvod}
m^{(1)} =\left\{\begin{array}{ll}\hat m b_U^{-1}, & k\in\Omega^U,\\
\hat m b_L^{-1}, & k\in\Omega^L,\\
\hat m, & \text{else}.
\end{array}\right.
\eeq
Note, that this deformation respects our symmetry condition \eqref{eq:symcond}.
Evidently, the matrices $b_L$ and $b_U$ have jumps on $\Sigma_a$. The new jump matrices $v^{(1)}(k)$, that correspond to $m^{(1)}(k)$ on this contour, are
\beq\label{jump5}
v^{(1)}(k)=\left\{\begin{array}{ll}
\begin{pmatrix}
\E^{\I t (g_+ - g_-)} & 0 \\
(R_- - R_+ +\chi)\Lambda^{-2} \E^{t(g_+ + g_-)} & \E^{-\I t (g_+ - g_-)}
\end{pmatrix},& k\in\Sigma_a^U,\\
\ & \ \\
\begin{pmatrix}
\E^{\I t (g_+ - g_-)} & (\ol R_- -\ol R_++\chi)\Lambda^{2} \E^{t(g_+ + g_-)} \\
0 & \E^{-\I t (g_+ - g_-)}
\end{pmatrix},& k\in\Sigma_a^L,
\end{array}\right.
\eeq
Again Lemma~\ref{lemzero} shows that the off-diagonal entries vanish.

Now set $\Sigma_{ac}=\Sigma_c\setminus\Sigma_a$, that is
\[
\Sigma_{ac}=\Sigma_{ac}^U\cup\Sigma_{ac}^L= [\I c, \I a]\cup [-\I a,-\I c].
\]
After the deformation \eqref{razvod} the jump along the
real axis disappears. Taking into account property {\bf (c)} of Lemma~\ref{KM} we obtain a new Riemann--Hilbert problem
\beq\label{defm1}
m_+^{(1)}(k)= m_-^{(1)}(k) v^{(1)}(k),\qquad m^{(1)}(k)\to(1,1),\ \  k\to\infty,
\eeq
where
\beq\label{jumpcondm1}
v^{(1)}(k) = \left\{ \begin{array}{ll} \hat{v}(k),\qquad  k\in \T_j^U\cup\T_j^L,
&j=1,...,N, \\
\ & \ \\
\begin{pmatrix}
\E^{\I t (g_+ - g_-)} & 0 \\
\Lambda^{-2}(k)\chi(k)  & \E^{-\I t (g_+ - g_-)}
\end{pmatrix},& k\in\Sigma_{ac}^U,\\
\ & \ \\
\begin{pmatrix}
\E^{\I t (g_+ - g_-)} & \Lambda^{2}(k)\chi(k) \\
0 & \E^{-\I t (g_+ - g_-)}
\end{pmatrix},& k\in\Sigma_{ac}^L,\\
\begin{pmatrix}\E^{-\I t B}& 0\\
0&\E^{\I t B}\end{pmatrix},& k\in\Sigma_a,\\
\ & \ \\
\begin{pmatrix} 1&0\\
\Lambda^{-2}(k)R(k) \E^{2\I t g(k)} & 1
\end{pmatrix}, & k\in\Sigma^U,\\
\ & \ \\
\begin{pmatrix} 1 & - \Lambda^2(k)\ol{R(k)} \E^{-2\I t g(k)} \\
0&1\end{pmatrix}, & k\in\Sigma^L.
\end{array}\right.
\eeq
Note that $\im g(k)>0$ (resp., $\im g(k)<0$) on the contour $\Sigma^U\setminus\{\I a\}$ (resp., $\Sigma^L\setminus\{-\I a\}$), and the corresponding matrices are exponentially close to the identity matrix except for small vicinities of the points $\pm\I a$.

Our next  step of conjugation deals with a factorization of the jump matrices on the set $\Sigma_{ac}.$ To this end consider an auxiliary scalar Riemann--Hilbert problem (cf. \cite{KM}):
Find a function $F(k)=F(k,\xi)$ analytic in the domain $\C\setminus\Sigma_c$ and a constant $\hat h(\xi)$ such that the following properties hold
\begin{itemize}
\item $F_+(k)F_-(k)=|\chi(k)|$ for $k\in\Sigma_{ac}^U=[\I c,\I a]$,
\item $F_+(k)F_-(k)=|\chi(k)|^{-1}$ for $k\in\Sigma_{ac}^L=[-\I a,-\I c]$,
\item $F_+(k)=F_-(k)\hat h $ for $k\in\Sigma_a=[\I a, - \I a]$,
\item $F(k)\to 1$ as $k\to \infty$ and $F(-k)=F^{-1}(k)$ for $k\in\C\setminus \Sigma_c$.
\end{itemize}
Note, that the last property allows us to use the function $F$ as an entry of the diagonal matrix for a conjugation step.

We construct the function $F$ using the Plemelj formulas. In the domain $\C\setminus\Sigma_{ac}$, introduce the function
\beq\label{koren}
w(k)=\sqrt{(k^2+c^2)(k^2 + a^2)},\quad w(0)>0,
\eeq
and for $k\in\Sigma_c$ set $p(k):= w(k)_+=w(k)_r$. Then
\beq\label{odd}
p(-k)=-p(k) \ \mbox{for}\  k\in\Sigma_{ac},\ \ \  p(k)=p(-k),\ \ \mbox{ as}\ k\in\Sigma_a.
\eeq
Set also \beq\label{deff}f(k):=\frac{\log |\chi(k)|}{p(k)}.
\eeq
Taking logarithms of the jump conditions and dividing them by $p(k)$ we get
\beq\label{forF}
F(k)=\exp\left\{\frac{w(k)}{2\pi\I}\left(\int_{\I c}^{\I a}\frac{f(s)}{s-k}ds +\int_{-\I c}^{-\I a}\frac{f(s)}{s-k}ds -\log \hat h \int_{-\I a}^{\I a}\frac{ds}{w(s)(s-k)}\right)\right\}.
\eeq
Properties \eqref{proptau} and \eqref{odd} imply
$F(-k)=F^{-1}(k)$. From this property, decomposing the function in exponent with respect to $k$ at infinity we conclude that
\beq\label{decompF}
F(k)=1 + \frac{y(\xi)}{\I k}+ O\left(\frac{1}{k^3}\right),\
\eeq
\beq \label{defy}
y(\xi)=\frac{1}{2\pi}\left\{- 2\int_{\I c}^{\I a}\frac{s^2\log|\chi(s)|}{w_+(s)}ds +\I\Delta\int_{\I a}^{-\I a}\frac{s^2 ds }{w(s)}\right\}\in\R,
\eeq
where $ w(k)=\sqrt{(k^2+c^2)(k^2 + a^2)},$
\beq \label{defdelta}
\Delta(\xi)=2\I\ \frac{\int_{\I a}^{\I c}\frac{\log|\chi(s)|}{w_+(s)}ds}{\int_{-\I a}^{\I a}\frac{ds}{w(s)}}\in\R,\ \mbox{and}\  \hat h(\xi)=\E^{\I\Delta(\xi)}.
\eeq
Now we are ready to perform the next deformation-conjugation step. Introducing $\tilde F(k)=F(k)\Lambda^{-1}(k)$ we get again a function satisfying $\tilde F(-k) = \ti F^{-1}(k)$ and hence the symmetry conditions of Lemma~\ref{lem:conjug}. Moreover, observe $\chi/\I=|\chi|$ on $\Sigma_{ac}^U$  $\chi/\I=-|\chi|$ on $\Sigma_{ac}^L$ by \eqref{proptau1}. Using also condition {\bf(b)} of Lemma~\ref{KM} one can  check, that on the contour $\Sigma_{ac}^U$ the jump matrix $v^{(1)}(k)$ can be factorized as
\[
v^{(1)}(k)= D_{2,-}
\begin{pmatrix} 1&\frac{\tilde F_-^2\Lambda^2\E^{-2\I t g_-}}{\chi}\\ 0&1 \end{pmatrix}
\begin{pmatrix}0&\I\\ \I & 0\end{pmatrix}
\begin{pmatrix} 1&\frac{\tilde F_+^2\Lambda^2\E^{-2\I t g_+}}{\chi}\\ 0&1 \end{pmatrix}
D_{2,+}^{-1},
\]
and on the contour $\Sigma_{ac}^L$
\[
v^{(1)}(k)= D_{2,-}
\begin{pmatrix} 1&0\\ \frac{\E^{2\I t g_-}}{\chi\Lambda^2\tilde F_-^2}&1 \end{pmatrix}
\begin{pmatrix}0&-\I\\ -\I & 0\end{pmatrix}
\begin{pmatrix} 1&0\\ \frac{\E^{2\I t g_+}}{\chi\Lambda^2\tilde F_+^2}&1 \end{pmatrix}
D_{2,+}^{-1}.
\]
where
\beq\label{newD}
D_2(k)=\begin{pmatrix} \tilde F^{-1}(k)&0\\ 0& \Tilde F(k)\end{pmatrix}.
\eeq
From \eqref{tau} we conclude, that
\[
v^{(1)}(k)=\left\{\begin{array} {ll} D_{2,-} G_-^U(k)\begin{pmatrix}0&\I\\ \I & 0\end{pmatrix} G_+^U(k)^{-1} D_{2,+}^{-1},& \ k\in\Sigma^U_{ac},\\
\ & \ \\
D_{2,-} G_-^L(k)\begin{pmatrix}0&-\I\\ -\I & 0\end{pmatrix} G_+^L(k)^{-1} D_{2,+}^{-1}, & k\in\Sigma^L_{ac},
\end{array}\right.
\]
where
\beq\label{defGG}
G^U(k)= \begin{pmatrix} 1 & \frac{\tilde F^2  \Lambda^2 \E^{-2\I t g}}{V}\\ 0 & 1 \end{pmatrix},\quad
G^L(k)= \begin{pmatrix} 1 & 0\\ \frac{ \E^{2\I  t g}}{V\tilde F^2 \Lambda^2} & 1 \end{pmatrix}
\eeq
with $V(k):=\ol {T(k)} T_1(k)=\frac{-4 k_1 k}{|W(k)|^2}$ for $k\in \mathbb C_c^U$ and $V(k)=V(-k)$ for $k\in \mathbb C_c^L$.

Introduce the symmetric domains $\Omega^U_{1}$ and $\Omega^L_{1}$ as depicted in Figure~\ref{fig3}.
Their boundary contours are oriented top-down.

\begin{figure}[h]
\begin{picture}(8,4.2)
\put(4,0){\line(0,1){4}}
\put(4,3.5){\vector(0,-1){0.1}}
\put(4,1.7){\vector(0,-1){0.1}}
\put(4,0.5){\vector(0,-1){0.1}}

\put(4.3,2.8){$\I a$}
\put(4.2,1){$-\I a$}
\put(4,1){\circle*{0.1}}
\put(4,3){\circle*{0.1}}

\put(4.2,4){$\I c$}
\put(4.1,-0.1){$-\I c$}
\put(4,0){\circle*{0.1}}
\put(4,4){\circle*{0.1}}

\put(4.1,2.1){$\Sigma_a$}

\put(2,2.5){$\Sigma^U$}
\put(2,1.2){$\Sigma^L$}
\put(5.5,2.5){$\Sigma^U$}
\put(5.5,1.2){$\Sigma^L$}

\put(3.4,3.3){$\Sigma^U_{ac}$}
\put(3.32,0.4){$\Sigma^L_{ac}$}

\put(4.92,4){$\Sigma^U_1$}
\put(4.97,-0.2){$\Sigma^L_1$}

\put(3.4,2.25){$\Omega^U$}
\put(3.4,1.5){$\Omega^L$}

\put(4.2,3.5){$\Omega_1^U$}
\put(4.2,0.3){$\Omega_1^L$}

\curve(4,3, 5,2.5, 7.5,2.3)
\curve(4,1, 5,1.5, 7.5,1.7)
\curve(4,3, 3,2.5, 0.5,2.3)
\curve(4,1, 3,1.5, 0.5,1.7)

\put(6.5,2.34){\vector(1,0){0.1}}
\put(6.5,1.66){\vector(1,0){0.1}}
\put(1.5,2.34){\vector(1,0){0.1}}
\put(1.5,1.66){\vector(1,0){0.1}}

\curve(4,3, 3,3.7, 4,4.3, 5,3.7, 4,3)
\curve(4,1, 3,0.3, 4,-0.3, 5,0.3, 4,1)

\put(4.8,3.42){\vector(-1,-1){0.1}}
\put(3.2,3.41){\vector(-1,1){0.1}}
\put(4.8,0.59){\vector(1,-1){0.1}}
\put(3.2,0.58){\vector(1,1){0.1}}

\curvedashes{0.05,0.05}
\curve(0.3,2, 7.7,2)

\end{picture}
  \caption{The second deformation step}\label{fig3}
\end{figure}
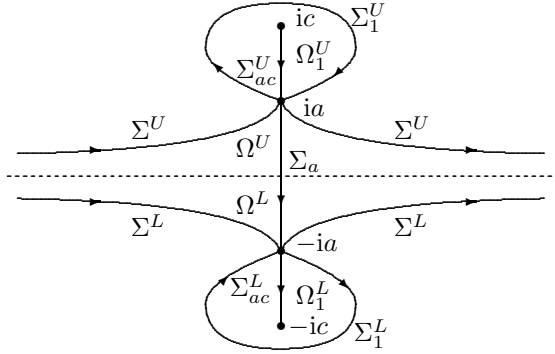

Introduce the new function
\beq\label{redef}
m^{(2)}(k)=m^{(1)}(k) G(k),\quad
G(k)=\left\{\begin{array}{ll} D_2(k) G^U(k), & k\in \Omega^U_{1},\\
D_2(k) G^L(k), & k\in\Omega^L_{1},\\
D_2(k), & \text{else}.\end{array}\right.
\eeq
Applying again Lemma~\ref{lem:conjug} we arrive at a new RH problem:
\beq\label{defm2}
m_+^{(2)}(k)= m_-^{(2)}(k) v^{(2)}(k),\quad m^{(2)}(k)\to(1,1),\ \  k\to\infty,
\eeq
where $m^{(2)}(-k)=m^{(2)}(k) \begin{pmatrix}0&1\\1&0\end{pmatrix}$ and
\beq\label{jumpcond25}
v^{(2)}(k) = \left\{ \begin{array}{ll}
\begin{pmatrix}
0 & \I \\
\I  & 0
\end{pmatrix},& k\in\Sigma_{ac}^U,\\
\begin{pmatrix}
0 & -\I \\
-\I & 0
\end{pmatrix},& k\in\Sigma_{ac}^L,\\
\begin{pmatrix}\E^{-\I t \hat B}& 0\\
0&\E^{\I t \hat B}\end{pmatrix},& k\in\Sigma_a,\\
G^U(k)^{-1}, & k\in\Sigma^U_1,\\
G^L(k)^{-1}, & k\in\Sigma^L_1,\\
D_2^{-1}(k)v^{(1)}(k)D_2(k), & k\in \cup_{j=1}^N(\T_j^U\cup\T_j^L)\cup\Sigma^U\cup\Sigma^L
\end{array}\right.
\eeq
Here $\hat B= B+\frac{\Delta}{t}$ (recall that $\tilde F_+\tilde F_-^{-1}=\E^{\I\Delta}$ on $\Sigma_a$, where $\Delta$ is defined by \eqref{defdelta}).

Note that due to \eqref{nonres}, \eqref{tau} and the definition of $F$ we get
\[
\lim_{k\to\I c} G^U_{12}\neq 0\ \ \mbox{and}\ \ \lim_{k\to -\I c} G^L_{21}\neq 0,
\]
and, therefore
\beq\label{asympm2}
m^{(2)}(k)= (c_1, c_2)(k\mp \I c)^{-1/4} (1+0(1)),\ \mbox{as}\ \ k\to\pm\I c,\ \ \mbox{where}\ \ c_1 c_2\neq 0.
\eeq
Now observe that on the contour
\[
\hat\Sigma:=\cup_{j=1}^N(\T_j^U\cup\T_j^L)\cup
\Sigma^U_1\cup\Sigma^L_1\cup
\Sigma^U\cup\Sigma^L,
\]
all jumps are exponentially close to the identity except for small vicinities of the points $\pm\I a$ as $t\to\infty$. To remove those parts one
needs to solve the corresponding RH problem corresponding to the jumps on $\Sigma^U_1\cup\Sigma^U$ restricted to a small neighborhood
of $+\I a$ (as well $\Sigma^L_1\cup\Sigma^L$ restricted to a small neighborhood of $-\I a$ which however follows from the first by symmetry).
Following the arguments from \cite{dkmvz} one can show that the contributions of these neighborhoods are negligible, that is,
\[
q(x,t) = q_0(x,t) + o(1),
\]
where $q_0(x,t)$ is obtained from a "model" RH problem, where all jumps on $\hat\Sigma$ are discarded. We will solve this model RH problem in
the next section.

\section{Solution of the model RH problem}
\label{sec:mp}

Consider the two-sheeted Riemann surface $X$ associated with the function $w(k)$, defined by \eqref{koren}, where we choose the standard branch of $\sqrt\cdot$ with the cut along the negative axis. The sheets of $X$ are glued along the cuts  $[\I c,\I a]$ and $[-\I a, -\I c]$. Points on this surface are denoted by $p=(k,\pm )$.
The canonical homology basis of cycles $\{\bf a, \bf b\}$ is chosen as follows: The $\bf a$-cycle surrounds the points $-\I a,\I a$ starting on the upper sheet from the left side of
the cut $[\I c,\I a]$ and continues on the upper sheet to the left part of $[-\I a, -\I c]$ and returns after changing sheets. The cycle $\bf b$ surrounds the points $\I a, \I c$ counterclockwise on the upper sheet.
Moreover, consider the normalized holomorphic differential
\beq\label{omm}
d\om=2\pi\I\frac{dk}{w(k)}\left(\int_{\bf a} \frac{dk}{w(k)}\right)^{-1},
\eeq
then $\int_{\bf a}d\om=2\pi\I,$ $\tau=\tau(\xi)=\int_{\bf b} d\om<0$. Let
\[
\theta(z)=\sum_{m\in\Z}\exp\left\{\frac 1 2 \tau m^2 + mz\right\}, \quad z\in\C
\]
be the  theta function and recall that $\theta$ is an even function, $\theta(-z)=\theta(z)$, satisfying
\[
\theta(z+2\pi\I n + \tau(\xi)\ell)=\theta(z)\exp\left\{-\frac 1 2 \tau(\xi) \ell^2 - \ell z\right\}.
\]
Furthermore, let $A(p)=\int_{\I c}^p d\om$ be the Abel map on $X$. Note that on
the upper sheet, where $p=(k,+)$, it has the following properties:
\begin{itemize}
\item $A_+(p)=-A_-(p)( \mod 2\pi\I)$ for $p\in\Sigma_{ac}$;
\item $A_+(p)-A_-(p)=-\tau$ as $p\in\Sigma_a$;
\item $A(-p)=-A(p) + \pi \I(\mod 2\pi\I)$ as $k\in\C\setminus\Sigma_c$, $p=(k,+)$;
\item $\ A(\I a)=-\frac{\tau}{2} (\mod\tau)$, $A(-\I a) = -\frac{\tau}{2} - \pi\I (\mod\tau,\mod 2\pi\I);$
\item $A((\infty,+))=\frac{\pi\I}{2}.$
\end{itemize}
Finally, denote by $K= \frac{\tau}{2}+\pi\I$ the Riemann constant associated with $X$.

Identifying  the upper sheet of $X$ with the complex plane we introduce two functions
\beq\label{defps1}
\alpha^b(k)=\theta\left(A(k) +\frac{\tau}{2}-K- \frac{\I t b }{2}\right)\theta\left(A(k) +\frac{\tau}{2}+\pi\I-K- \frac{\I t b }{2}\right),\eeq
\beq\label{defps2}\beta^b(k)=\theta\left(-A(k) +\frac{\tau}{2}-K- \frac{\I t b }{2}\right)\theta\left(-A(k) +\frac{\tau}{2}+\pi\I-K- \frac{\I t b }{2}\right),
\eeq
where $b\in\R$ will be determined later and $A(k)= A((k,+))$ for $k\in\C$.

Evidently, both functions $\alpha^0$ and $\beta^0$ have zeros of order one (on $X$) at the points $\pm \I a$. Moreover,
\beq\label{psiinf}
 \lim_{k\to\infty}\alpha^b(k)=
 \lim_{k\to\infty}\beta^b(k)=\theta\left(\frac{\pi\I}{2} +\frac{\I t b}{2}\right)\theta\left(-\frac{\pi\I}{2} +\frac{\I t b}{2}\right).
\eeq
Due to the first three properties of the Abel map we get
\beq\label{conj8}
\alpha^b_+(k)=\beta^b_-(k)\ \mbox{and}\  \beta^b_+(k)=\alpha_-^b(k)\ \mbox{for}\ k\in\Sigma_{ac}=(\Sigma_{ac}^U\cup\Sigma_{ac}^L).
\eeq
\beq\label{conj9}
 \frac{\alpha^b_+(k)}{\alpha^0_+(k)}
 =\E^{-\I b t}\frac{\alpha^b_-(k)}{\alpha^0_-(k)}\ \mbox{and}\ \frac{\beta^b_+(k)}{\beta^0_+(k)}=\E^{\I b t}\frac{\beta^b_-(k)}{\beta^0_-(k)}\ \mbox{for}\ k\in\Sigma_{a},
\eeq
\beq\label{conj10} \alpha^b(-k)=\beta^b(k) \ \mbox{for}\ k\in\C\setminus \Sigma_{c}.
\eeq
Now introduce the function
\beq\label{defgamma}
\gamma(k)=\sqrt[4]{\frac{k^2 + a^2}{k^2+c^2}},
\eeq
defined uniquely on the set $\C\setminus\Sigma_{ac}$ by the condition $\arg\gamma(0)=0$. This function satisfy the jump conditions
\beq\label{jumpga}
\begin{array}{ll} \gamma_+(k)=\I\gamma_-(k), & k\in\Sigma_{ac}^U\\
\gamma_+(k)=-\I\gamma_-(k), & k\in\Sigma_{ac}^L.
\end{array}
\eeq
Combining \eqref{psiinf}--\eqref{jumpga} we conclude that the vector
\beq\label{mmod}
m^{(3)}(k)=
\left(\gamma(k)\frac{\alpha^{\hat B}(k)\alpha^0(\infty)}{\alpha^0(k)\alpha^{\hat B}(\infty)},\ \gamma(k)\frac{\beta^{\hat B}(k)\beta^0(\infty)}{\beta^0(k)\beta^{\hat B}(\infty)}\right)
\eeq
solves our model problem
\beq\label{defm3}
m_+^{(3)}(k)= m_-^{(3)}(k) v^{(3)}(k),\quad m^{(3)}(k)\to(1,1),\ \  k\to\infty,
\eeq
where
\beq\label{jumpcond55}
v^{(3)}(k) = \left\{ \begin{array}{ll}
\begin{pmatrix}
0 & \I \\
\I  & 0
\end{pmatrix},& k\in\Sigma_{ac}^U,\\
\begin{pmatrix}
0 & -\I \\
-\I & 0
\end{pmatrix},& k\in\Sigma_{ac}^L,\\
\begin{pmatrix}\E^{-\I t \hat B}& 0\\
0&\E^{\I t \hat B}\end{pmatrix},& k\in\Sigma_a,\\
\end{array}\right..
\eeq
The symmetry condition
\beq\label{sym3}
m^{(3)}(-k)=\begin{pmatrix}0&1\\1&0\end{pmatrix}m^{(3)}(k),\quad k\in\C\setminus\Sigma_c
\eeq
is also fulfilled due to \eqref{conj10}.

Moreover, both components of the vector-valued function $m^{(3)}(k)$ are bounded everywhere except for small vicinities of the points $\pm\I a$, $\pm\I c$,
where they have singularities of the type $(k-\zeta)^{-1/4}$, $\zeta\in\{\I c, \I a, -\I c,-\I a\}$.

In summary we get
\begin{align}\label{m1}
 m_1^{(3)}(k) &= \sqrt[4]{\frac{k^2 + a^2}{k^2+c^2}}\frac{\theta\left(A(k) - \I\pi-\frac{\I t\hat B}{2}\right)\theta\left(A(k) -\frac{\I t\hat B}{2}\right)\theta^2\left(\frac{\pi\I}{2}\right)}
{{\theta\left(A(k) - \I\pi\right)\theta\left(A(k) \right)\theta\left(\frac{\pi\I}{2}-\frac{\I t\hat B}{2}\right)\theta\left(\frac{\pi\I}{2}+\frac{\I t\hat B}{2}\right)}},\\
\label{m2}
m_2^{(3)}(k) &= \sqrt[4]{\frac{k^2 + a^2}{k^2+c^2}}\frac{\theta\left(-A(k) - \I\pi-\frac{\I t\hat B}{2}\right)\theta\left(-A(k) -\frac{\I t\hat B}{2}\right)\theta^2\left(\frac{\pi\I}{2}\right)}
{{\theta\left(-A(k) - \I\pi\right)\theta\left(-A(k) \right)\theta\left(\frac{\pi\I}{2}-\frac{\I t\hat B}{2}\right)\theta\left(\frac{\pi\I}{2}+\frac{\I t\hat B}{2}\right)}}.
\end{align}
We are interested in the terms of order $\frac{1}{k}$ as $k\to+\I\infty$. To this end let
\beq\label{ellipt1}
\Gamma=\Gamma(\xi)=-\pi\left(\int_{-a}^a\frac{ds}{\sqrt{(c^2 - s^2)(a^2 - s^2)}}\right)^{-1}<0
\eeq
be the normalizing constant from the Abel integral.
Since $w(k)=k^2 ( 1+o(1))$ as $k\to +\I\infty$, we infer
\[
A(k) - A(+\infty)= A(k) - \frac{\pi\I}{2}=-\frac{\Gamma}{ k}+O\left(\frac{1}{k^2}\right)
\]
and
\[
\frac{\theta\left( \frac{\pi\I}{2}\right)}{\theta(A(k))}= 1
+\frac{\Gamma}{k}\frac{d}{d u}\log\theta (u)\mid_{u=
\frac{\pi\I}{2}}
+O\left(\frac{1}{k^2}\right).
\]
Proceeding in the same way with the other theta functions and taking into account that $\gamma(k)=1+O(k^{-2})$ for large $k$, we get
\[
m_1^{(3)}(k)=1 +\frac{\left(\hat E\left(
\frac{\pi\I}{2}\right)-
\hat E\left(\frac{\pi\I}{2}-\frac{\I t\hat B}{2}\right)\right)}{ k} + O\left(\frac{1}{k^2}\right),
\]
\[
m_2^{(3)}(k)=1 -\frac{\left(\hat E\left(
\frac{\pi\I}{2}\right)-
\hat E\left(\frac{\pi\I}{2}-\frac{\I t\hat B}{2}\right)\right)}{ k} + O\left(\frac{1}{k^2}\right),
\]
where
\[
\hat E(u)=\hat E(u,\xi)=\Gamma
\frac{d}{d u}\log\left(\theta(u)\theta(u-\I\pi)\right).
\]
To check that these asymptotics are real-valued on the imaginary axis, recall that
$\theta(u)=\theta_3\left(\frac{u}{2\pi \I}\right)$, where (cf.\ \cite{Akh})
\[
\theta_3(v)=\theta_3(v\mid\tau_1)=\sum_{m\in\Z}\exp\{(m^2\tau_1 + 2mv)\pi\I\},\quad \tau_1=\tau_1(\xi)=\frac{\tau(\xi)}{2\pi\I}\in \I\R_+.
\]
Then $\frac{d}{du}\theta(u)=\frac{1}{2\pi \I}\frac{d}{dv}\theta_3(v)\mid_{v=\frac{u}{2\pi \I}}$
and
\[
m^{(3)}(k)=(1,1) +\frac{E\left(
\frac{1}{4}\right)-
E\left(\frac{1}{4} - \frac{ t B + \Delta}{4\pi}\right)}{ \I k}(1,-1) + O\left(\frac{1}{k^2}\right),
\]
where
\beq\label{thetalast}
E(v)=E(v,\xi)=\frac{\Gamma}{2\pi}\frac{d}{dv}
\log\left(\theta_3(v)\theta_3(v-\frac{1}{2})\right)
\eeq
is a real-valued function for $v\in\R$.

Since outside small vicinities of the points $\pm\I c, \pm \I a$ the solution of the model problem
$m^{(3)}(k)$ approximates the solution $m^{(2)}(k)$ of the problem \eqref{defm2}--\eqref{jumpcond25} it remains
to trace back our deformation and conjugation steps:
\[
m(k)\begin{pmatrix}\Lambda^{-1}(k)&0
\\0&\Lambda(k)\end{pmatrix}
\begin{pmatrix}d^{-1}(k,t)&0\\0&d(k,t)\end{pmatrix}
\begin{pmatrix}\Lambda(k)F^{-1}(k)&0\\0&\Lambda^{-1}(k) F(k)\end{pmatrix}
=m^{(2)}(k),
\]
where the asymptotic behavior of $d(k,t)$ and $F(k)$ are given by
\eqref{asdt} and \eqref{decompF}--\eqref{defdelta}.
Therefore, with $\frac{x}{12 t}=\xi$, we have
\beq\label{integr}
m(k,x,t)=(1,1) +\left(\frac{t z(\xi)-
E\left(\frac{1}{4} -  \frac{t B + \Delta}{4\pi}\right)}{ \I k}+\frac{E\left(
\frac{1}{4}\right)+y(\xi)}{\I k}\right)(1,-1) + O\left(\frac{1}{k^2}\right).
\eeq
This formula together with \eqref{asm} gives us the asymptotic behavior of $\int_x^\infty q(s,t)ds$.
To get the asymptotic behavior of $q(x,t)$ we differentiate \eqref{integr} with respect to $x$, taking into account that for any smooth function $p(\xi)$ one has
$\frac{d}{dx}p(\xi)=O\left(\frac{1}{t}\right)$. Thus we obtain
\beq\label{final}
q(x,t)=-\frac{1}{24\pi}E'\left(\frac{1}{4} -  \frac{t B(\xi) + \Delta(\xi)}{4\pi}\right) B'(\xi) + \frac{1}{6} z'(\xi) +
O\left(\frac{1}{t}\right),
\eeq
where
\beq\label{thetaf}
E^\prime(v)=-\tilde \Gamma(\xi)\frac{d^2}{dv^2}
\log\left(\theta_3\left(v\mid\tau_1(\xi)\right)
\theta_3\left((v-\frac{1}{2})\mid\tau_1(\xi)\right) \right)
\eeq
and
\beq\label{Gamdef}
\tilde\Gamma(\xi)=
\frac{1}{2}\left(\int_{-a(\xi)}^{a(\xi)}
\left((c^2 - s^2)(a^2(\xi) - s^2)\right)^{-1/2}ds\right)^{-1}.
\eeq
Formula \eqref{final} can be simplified using the following formula for summing theta-functions (\cite{dubr} formula (1.4.3))
 \[
 \theta_3(z+w\mid\frac{\tau}{2})
 \theta_3(z-w\mid\frac{\tau}{2})=\theta_3(2z\mid \tau)\theta_3(2w\mid \tau)+\theta_2(2z\mid \tau)\theta_2(2w\mid \tau),
 \]
 where
 \[
 \theta_2(z \mid\tau)=\sum_{m\in\Z}\exp\{\pi\I (m+ \frac 1 2)^2\tau + 2\pi\I (m +\frac 1 2) z\}.
 \]
 Since $\theta_2(\frac 1 2\mid\tau)=0$, we see
 \[
 \theta_3(u\mid\frac{\tau}{2})
 \theta_3(u-\frac 1 2\mid\frac{\tau}{2})=\theta_3(2u-\frac 1 2\mid \tau)\theta_3(\frac 1 2\mid \tau)
 \]
and the last formula implies that
\[
 \log\left(\theta_3\left(v\mid\tau_1(\xi)\right)
\theta_3\left((v-\frac{1}{2})\mid\tau_1(\xi)\right) \right)=\log\theta_3(2v-\frac 1 2\mid 2\tau_1(\xi)) +f(\xi).
\]
Substituting $v=\frac{1}{4}-\frac{tB(\xi)+\Delta(\xi)}{4\pi}$ and taking into account \eqref{final} as well as the estimate $\frac{d}{dx}f(\xi)=O(\frac{1}{t})$ we get

\begin{theorem}\label{thm:ellipt}
Assume \eqref{decay}--\eqref{decay1}. Then in the domain $-6c^2+\varepsilon<\frac{x}{t}<4c^2 - \varepsilon_1$
the following asymptotical formula is valid
\begin{align}\nn
 q(x,t)=&\frac{\tilde\Gamma(\xi)}{6\pi}\frac{d^2}{dv^2}
 \log\theta_3\left(\frac{t B(\xi) +\Delta(\xi)}{2\pi}+v \right)\mid_{v=0}
 \frac{d}{d\xi}B(\xi)+ \\ \label{secon}
 &+\frac{1}{6}\frac{d}{d\xi}z(\xi) + o(1).
\end{align}
Here $\theta_3(v)=\theta_3(v\mid \tau(\xi))$ with $\tau(\xi)=\frac{1}{\pi\I}\int_{\bf b} d\om$, where $d\omega$ is the normalized holomorphic differential \eqref{omm}.
The function $\tilde\Gamma(\xi)>0$ is defined by \eqref{Gamdef} and
\begin{align*}
B(\xi) &=12\int_{a(\xi)}^{c}\left(\xi +\frac{c^2 -a^2(\xi)}{2} - s^2\right)\sqrt{\frac{s^2 - a^2(\xi)}{c^2 - s^2}}ds,\\
z(\xi) &=\frac{12\xi(c^2 - a(\xi)^2) + 3c^4 + 9a(\xi)^4-6a(\xi)^2c^2}{2},\\
\Delta(\xi) &=2\int_{a(\xi)}^c\frac{\log|(\ol T(\I s)T_1(\I s)|}{\sqrt{(c^2 - s^2)(s^2 -a(\xi)^2 )}} ds\left(\int_{-a(\xi)}^{a(\xi)}\frac{ds}{\sqrt{(c^2 - s^2)(a^2(\xi)-s^2)}}\right)^{-1}
\end{align*}
are real-valued functions. In all formulas for $\tilde\Gamma$, $B$, $z$ and $\Delta$ the positive values of the square roots are taken.
\end{theorem}

\section{Asymptotics in the domain $x<-6c^2 t$.}
\label{sec:dr}

To study the asymptotical behavior of $q(x,t)$ in the domain we use the RH problem, associated with the left half axis. Namely, we consider the spectral data and the Jost solutions as the functions of the parameter $k_1=\sqrt{k^2 + c^2}$. Then the continuous spectrum of the operator $H(0)$ coincides with the set $\im(k_1)=0$
and the discrete spectrum is located at the points $\I\kappa_{1,j}=\I\sqrt{\kappa_j^2 - c^2}$ (recall that $\kappa_j^2>c^2$). Introduce the vector-valued function
\beq\label{defmn}
m(k_1,x,t)= \left\{\begin{array}{c@{\quad}l}
\begin{pmatrix} T_1(k_1,t) \phi(k_1,x,t) \E^{-\I k_1 x}  & \phi_1(k_1,x,t) \E^{\I k_1 x} \end{pmatrix},
& k_1\in \C^U, \\
\begin{pmatrix} \phi_1(-k_1,x,t) \E^{-\I k_1 x} & T_1(-k_1,t) \phi(-k_1,x,t) \E^{\I k_1 x} \end{pmatrix},
& k_1\in\C^L,
\end{array}\right.
\eeq
where $\C^U:=\{k_1:\ \im(k_1)>0\}$,
$\C^L:=\{k_1:\ \im(k_1)<0\}$.

This function has the following asymptotical behavior
\beq\label{asm1} m(k_1,x,t)= (1,1) +\frac{1}{2\I k_1}\left(\int^x_{-\infty}(q(y,t)+ c^2)dy\right) (1,-1) + O\left(\frac{1}{k_1^2}\right).\eeq
\begin{theorem}\label{thm:vecrhpn}
Let $\{ R_1(k_1),\; k_1\in \R; \  (\kappa_{1,j}, \ga_{1,j}), \: 1\le j \le N \}$ be
the left scattering data of the operator $H(0)$. Let $\T_j^U$ (resp., $\T_j^L$) be circles with centers in $\I \kappa_{1,j}$ (resp., $-\I \kappa_{1,j}$) and radiuses $0<\varepsilon<\frac{1}{4}\min_{j=1}^N|\kappa_{1,j} - \kappa_{1,j-1}|,$ $\kappa_{1,0}=0$. Then $m(k_1)=m(k_1,x,t)$ defined in \eqref{defmn}
is a solution of the following vector Riemann--Hilbert problem.

Find a function $m(k_1)$ which is holomorphic away from the contour $\cup_{j=1}^N (\T_j^U\cup\T_j^L)\cup\R$  and satisfies:
\begin{enumerate}
\item The jump condition $m_+(k_1)=m_-(k_1) v(k_1)$
\beq \label{jumpcond88}
v(k_1)=\left\{\begin{array}{cc}\begin{pmatrix}
1-|R_1(k_1)|^2 & - \ol{R_1(k_1)} \E^{-t\Phi_1(k_1)} \\
R_1(k_1) \E^{t\Phi_1(k_1)} & 1
\end{pmatrix},& k_1\in\R\setminus [- c,  c]\\
 \ &\ \\
\begin{pmatrix}
0 & - \ol{R_1(k_1)} \E^{-t\Phi_1(k_1)} \\
R_1(k_1) \E^{t\Phi_1(k_1)} & 1
\end{pmatrix},& k_1\in [- c,  c]\\
\ &\ \\
 \begin{pmatrix} 1 & 0 \\
-\frac{\I \gamma_{1,j}^2 \E^{t\Phi_1(\I\kappa_{1,j})}}{k_1-\I \kappa_{1,j}} & 1 \end{pmatrix},& k_1\in \T_j^U,\\
\ &\ \\
 \begin{pmatrix} 1 & -\frac{\I \gamma_{1,j}^2 \E^{-t\Phi_1(-\I \kappa_{1,j})}}{k_1+ \I \kappa_{1,j}} \\
0 & 1 \end{pmatrix},& k_1\in\T_j^L,
\end{array}\right.
\eeq
\item
the symmetry condition
$
m(-k_1) = m(k_1) \sigI,
$
\item
the normalization condition
$
\lim_{\kappa\to\infty} m(\I\kappa) = (1\quad 1).
$
\end{enumerate}
Here the phase $\Phi_1(k)=\Phi_1(k_1,x,t)$ is given by
\begin{equation}\label{phi1}
\Phi_1(k_1)= -8 \I k_1^3+12\I c^2 k_1 -24\I\xi k_1,\quad \xi=\frac{x}{12 t}.
\end{equation}
\end{theorem}

\begin{proof}
The proof of this theorem is similar to the proofs of Theorem \ref{thm:vecrhp} and Lemma \ref{lem:holrhp}.
It is based on formulas \eqref{posl}, \eqref{realos}, \eqref{ident}, the formula $\phi(-k_1,x,t)=\phi(k_1,x,t)\in\R$ for $k_1\in[-c,c]$
and the relations (cf.\ \cite{EGT})
\[
R_1(k_1,t)=R_1(k_1,0)\E^{-4\I t(-c^2 + 2k^2)k_1}, \qquad \gamma_{1,j}^2(t)=\gamma_{1,j}^2(0)\E^{-8\kappa_j^2 \kappa_{1,j} t -4 c^2\kappa_{1,j}t}.
\]
\end{proof}

Denote by $ \pm k_{1,0}=\pm \sqrt{\frac{c^2}{2}-\xi}$ the stationary phase points of $\Phi_1$, that is,
the zeros of the equation $\Phi_1^\prime(k_1)=0$. In the present domain $\xi<-\frac{c^2}{2}$
we have  $k_{1,0}>c$ and the signature table for $\re \Phi_1$ is shown in
Figure~\ref{fig4}.

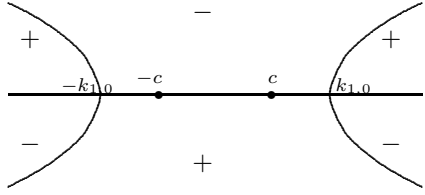
\begin{figure}[h]
\begin{picture}(6,3)
\put(0,1.25){\line(1,0){5.5}}

\put(0.15,0.5){$-$}
\put(0.15,1.9){$+$}
\put(4.95,0.5){$-$}
\put(4.95,1.9){$+$}
\put(2.45,2.25){$-$}
\put(2.45,0.25){$+$}

\put(0.7,1.3){$\scriptstyle -k_{1,0}$}
\put(4.35,1.3){$\scriptstyle k_{1,0}$}
\put(1.7,1.4){$\scriptstyle -c$}
\put(3.45,1.4){$\scriptstyle c$}
\put(2,1.25){\circle*{0.1}}
\put(3.5,1.25){\circle*{0.1}}

\curve(0.,0.025, 0.425,0.25, 0.775,0.5, 1.025,0.75, 1.225,1.25, 1.025,1.75, 0.775,2., 0.425,2.25, 0.,2.47)

\curve(5.5,0.025, 5.075,0.25, 4.725,0.5, 4.475,0.75, 4.275,1.25, 4.475,1.75, 4.725,2., 5.075,2.25, 5.5,2.47)
\end{picture}
  \caption{Sign of $\re(\Phi_1(k_1))$}\label{fig4}
\end{figure}

First of all we observe, that the jump matrices corresponding to the discrete spectrum are exponentially close to the identity matrices as $t\to\infty$.
Therefore, unlike in the previous cases we do not need a conjugation step for them.
Moreover, since the parameter $k$ will not appear in the remainder of this section, we will write $k$ in place of $k_1$ to simplify notations.

From \eqref{posl} it follows that $1-|R_1(k)|^2=0$ for $k\in[-c,c]$ and hence
\[
v(k)=\begin{pmatrix}
1-|R_1(k)|^2 & - \ol{R_1(k)} \E^{-t\Phi_1(k)} \\
R_1(k) \E^{t\Phi_1(k)} & 1
\end{pmatrix},\qquad k\in\R.
\]
Now following the usual procedure \cite{dz}, \cite{GT}  we let $d(k)$ be an analytic function in the domain $ \C\setminus\left(\R\setminus[-k_{1,0},k_{1,0}]\right)$ satisfying
\[
d_+(k)=d_-(k)(1-|R_1(k)|^2)\ \mbox{for}\ k\in\R\setminus[-k_{1,0}, k_{1,0}] \ \mbox { and}\ d(k)\to 1, \ k\to\infty.
\]
Then by the Plemelj formulas
\beq\label{defdel}
d(k)=\exp\left(\frac{1}{2\pi \I} \int_{\R\setminus[-k_{1,0}, k_{1,0}]} \frac{\log(1-|R_1(s)|^2)}{s-k} ds\right).
\eeq
For the smooth steplike initial data $q(x,0)\in C^n(\R)$ the reflection coefficient satisfies
$R_1(k)=O(k^{n+1})$ (cf. \cite{EGT}). Moreover, in the domain $\R\setminus [-k_{1,0}, k_{1,0}]$ we have $|R_1(k)|<1$. Therefore, the integral under the exponent is well defined.
Since the domain of integration here is even and the function $\log(1-|R_1|^2)$ is also even, we obtain $d(-k)=d^{-1}(k)$ and the matrix
\beq\label{newD1}
D(k)=\begin{pmatrix}d^{-1}(k)&0\\
0&d(k)\end{pmatrix}
\eeq
satisfies the symmetry conditions of Lemma~\ref{lem:conjug}. Now set $\tilde m(k)=m(k)D(k)$ and the new RH problem will read
$\tilde m_+(k)=\tilde m_-(k) \tilde v(k)$, where
$\tilde m(k)\to (1,1)$ as $k\to\infty$, $\tilde m(-k)=\tilde m(k) \left(\begin{smallmatrix}0&1\\1&0\end{smallmatrix}\right)$ and
\beq\label{conj75}
 \tilde v(k)=\left\{\begin{array}{ll}
 A^L_-(k)A^U_+(k), & k\in\R\setminus[-k_{1,0},k_{1,0}]\\
 \ & \\
 B^L(k)B^U(k), & k\in [-k_{1,0},k_{1,0}]\\
 \ & \\
D^{-1}(k)v(k) D(k), & k\in \cup_j(\T_j^U\cup\T_j^L),
\end{array}\right.
\eeq
where
\beq\label{razvod3}
A^L(k)=\begin{pmatrix}1&0\\ \frac{R_1(k) \E^{t\Phi_1(k)}}{(1-|R_1(k)|^2)d^2(k)} &1\end{pmatrix},\quad k\in\Omega_l^L\cup\Omega_r^L,
\eeq
\beq\label{razvod6}
A^U(k)=\begin{pmatrix} 1& -\frac{d^2(k)\ol {R_1(k)} \E^{-t\Phi_1(k)}}{(1-|R_1(k)|^2)}\\0&1\end{pmatrix}, \quad k\in \Omega_l^U\cup\Omega_r^U,
\eeq
\beq\label{razvod4}B^L(k)=\begin{pmatrix}1 & -d^2(k)\ol {R_1(k)}\E^{-t\Phi_1(k)}\\ 0 & 1\end{pmatrix},\quad k\in \Omega_c^L,
\eeq
\beq\label{razvod9}B^U(k)=\begin{pmatrix} 1& 0\\
d^{-2}(k) R_1(k)\E^{t\Phi_1(k)}& 1\end{pmatrix}, \quad k\in\Omega_c^U.
\eeq
Here the domains $\Omega_l^L$, $\Omega_l^U$,
$\Omega_c^L$, $\Omega_c^U$, $\Omega_r^L$, $\Omega_r^U$, are
bounded by the contours $\Sigma_l^L$, $\Sigma_l^U$, $\Sigma_c^L$, $\Sigma_c^U$, $\Sigma_r^L$, $\Sigma_r^U$, as shown in Figure~\ref{fig5}.
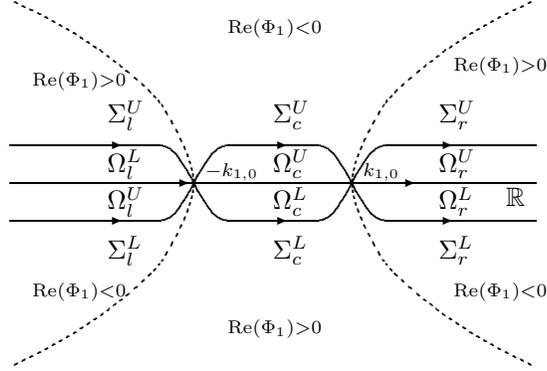
\begin{figure}[h]
\begin{picture}(7,5.2)
\put(0,2.5){\line(1,0){7.0}}
\put(2,2.5){\vector(1,0){0.4}}
\put(5,2.5){\vector(1,0){0.4}}

\put(6.6,2.2){$\R$}

\put(0,2){\line(1,0){2.0}}
\put(2.9,3){\line(1,0){1.2}}
\put(5,2){\line(1,0){2.0}}
\put(1.1,2){\vector(1,0){0.4}}
\put(5.5,2){\vector(1,0){0.4}}
\put(3.3,3){\vector(1,0){0.4}}

\put(1.3,1.5){$\Sigma_l^L$}
\put(1.3,2.65){$\Omega_l^L$}
\put(5.7,1.5){$\Sigma_r^L$}
\put(5.7,2.15){$\Omega_r^L$}
\put(3.5,3.3){$\Sigma_c^U$}
\put(3.5,2.65){$\Omega_c^U$}

\curve(2.,2., 2.2,2.1, 2.4,2.4, 2.45,2.5, 2.5,2.6, 2.7,2.9, 2.9,3.)

\curve(4.1,3., 4.3,2.9, 4.5,2.6, 4.55,2.5, 4.6,2.4, 4.8,2.1, 5.,2.)

\put(0,3){\line(1,0){2.0}}
\put(2.9,2){\line(1,0){1.2}}
\put(5,3){\line(1,0){2.0}}
\put(1.1,3){\vector(1,0){0.4}}
\put(5.5,3){\vector(1,0){0.4}}
\put(3.3,2){\vector(1,0){0.4}}

\curve(2.,3., 2.2,2.9, 2.4,2.6, 2.45,2.5, 2.5,2.4, 2.7,2.1, 2.9,2.)

\curve(4.1,2., 4.3,2.1, 4.5,2.4, 4.55,2.5, 4.6,2.6, 4.8,2.9, 5.,3.)

\put(1.3,3.3){$\Sigma_l^U$}
\put(1.3,2.15){$\Omega_l^U$}
\put(5.7,3,3){$\Sigma_r^U$}
\put(5.7,2.65){$\Omega_r^U$}
\put(3.5,1.5){$\Sigma_c^L$}
\put(3.5,2.15){$\Omega_c^L$}

\put(0.3,1.0){$\scriptstyle\re(\Phi_1)<0$}
\put(0.3,3.8){$\scriptstyle\re(\Phi_1)>0$}
\put(5.9,1.0){$\scriptstyle\re(\Phi_1)<0$}
\put(5.9,4.0){$\scriptstyle\re(\Phi_1)>0$}
\put(2.9,4.5){$\scriptstyle\re(\Phi_1)<0$}
\put(2.9,0.5){$\scriptstyle\re(\Phi_1)>0$}

\put(2.6,2.6){$\scriptstyle -k_{1,0}$}
\put(4.7,2.6){$\scriptstyle k_{1,0}$}

\curvedashes{0.05,0.05}

\curve(0.,0.05, 0.85,0.5, 1.55,1., 2.05,1.5, 2.45,2.5, 2.05,3.5, 1.55,4., 0.85,4.5, 0.,4.94)

\curve(7.,0.05, 6.15,0.5, 5.45,1., 4.95,1.5, 4.55,2.5, 4.95,3.5, 5.45,4., 6.15,4.5, 7.,4.94)
\end{picture}
\caption{Contour deformation in the dispersive region}\label{fig5}
\end{figure}
All contours are oriented from left to right. They are chosen to respect the symmetry $k\mapsto -k$ and are inside the strip $|\im k|<\min\{\frac{C_0}{2},\frac{\kappa_{1,1}}{4}\}$ below the discrete spectrum and inside the domain, where $R_1(k)$ has an analytic continuation. We also set $\ol{R_1(k)}=R_1(-k)$ in these domains.

Now redefine $\tilde m(k)$ according to
\beq\label{redef4}
\hat m(k)=\left\{\begin{array}{ll}
\tilde m(k)A^L(k), & k\in \Omega_l^L\cup\Omega_r^L,\\
 \tilde m(k) A^U(k)^{-1}, & k\in \Omega_l^U\cup\Omega_r^U,\\
 \tilde m(k)B^L(k), & k\in \Omega_c^L,\\
 \tilde m(k) B^U(k)^{-1}, & k\in \Omega_c^U,\\
 \tilde m(k), & \mbox{else.}\end{array}\right.
\eeq
Now the function $\hat m(k)$ has no jump as$k\in\R$ and all evidently defined jumps on contours $\Sigma_l^L$, $\Sigma_l^U$, $\Sigma_c^L$, $\Sigma_c^U$, $\Sigma_r^L$, $\Sigma_r^U$, $\cup_{j=1}^N(\T_j^U\cup\T_j^L)$ are exponentially small with respect to $t$ outside of small vicinities of the stationary phase  points $k_{1,0}$ and $-k_{1,0}$. Thus, the model problem here has the trivial solution $\hat m(k)=(1,1)$. For large imaginary $k$ with $|k|>\kappa_{1,N}+1$ we have $ \tilde m(k)=\hat m(k)$ and consequently
\[
m(k)=\tilde m(k)D^{-1}(k)=(d(k), d^{-1}(k))
\]
for sufficiently large $k$.
By \eqref{defdel}
\[
d(k_1)= 1 +\frac{1}{2\I k_1}\left(-\frac{1}{\pi}\int_{\R\setminus[-k_{1,0}, k_{1,0}]}\log(1-|R_1(s)|^2) ds\right) +O\left(\frac{1}{k_1^2}\right).
\]
Comparing this formula with formula \eqref{asm1} we now can derive the asymptotics using Theorem~A.1 from \cite{KTb} following literally
the argument in Section~5 of \cite{GT}:

\begin{theorem}\label{thm:asym2}
Assume  \eqref{decay} and \eqref{decay1}. Then the asymptotics in the similarity region, $\frac{x}{t}+ 6c^2 < -\varepsilon$ for some $\varepsilon>0$, are given by
\beq\label{eq:simasymp}
\aligned
q(x,t)= -c^2 + \sqrt{\frac{4\nu(k_{1,0}) k_{1,0}}{3t}}\sin(16tk_{1,0}^3-\nu(k_{1,0})\log(192 t k_{1,0}^3)+\delta(k_{1,0}))+O(t^{-\alpha})
\endaligned
\eeq
for any $1/2<\alpha <1$.
Here $k_{1,0}= \sqrt{\frac{c^2}{2}-\frac{x}{12t}}$ and
\begin{align}
\nu(k_{1,0})= & -\frac{1}{2\pi} \log(1-|R_1(k_{1,0})|^2),\\ \nn
\delta(k_{1,0})= & \frac{\pi}{4}- \arg(R_1(k_{1,0}))+\arg(\Gamma(\I\nu(k_{1,0})))\\
&  -\frac{1}{\pi}\int_{\R\setminus[-k_{1,0}, k_{1,0}]}\log\left(\frac{1-|R_1(\zeta)|^2}{1-|R_1(k_{1,0})|^2}\right)\frac{1}{\zeta-k_{1,0}}d\zeta.
\end{align}
\end{theorem}

\appendix

\section{Inverse scattering transform on steplike backgrounds}
\label{sec:app}

The purpose of this Appendix is to prove some facts from scattering theory used in this paper. Most of the properties listed here are valid for a much wider class of potentials then those satisfying \eqref{decay}, namely, for continuous potentials with a finite second moment:
\beq\label{second}
\int_0^{+\infty}(1+x^2)(|q(x,t)| + |q(-x,t)+c^2|)dx<\infty.
\eeq
We start with

\begin{proof}[Proof of Lemma \ref{asypm}]
We will omit the dependence on $t$ for notational simplicity.
Let $\phi(k,x)$ and $\phi_1(k,x)$ be the Jost solutions \eqref{phipl} of equation \eqref{shturm}, normalized by \eqref{lims}.
According to \eqref{rscat} the right transmission coefficient $T(k)$ is defined by formula \eqref{ident}. Our first step is to compute its asymptotics as $k\to \infty$ up to a term $o\left(\frac{1}{k}\right)$. Since the Wronskian \eqref{wronsk} does not depend on $x$, we evaluate it at $x=0$.
Under condition \eqref{second} the integrals in \eqref{phipl} can be integrated by parts one time and then differentiated with respect to $x$. We get
\begin{align*}
\phi_1(k,0) &=1 -\frac{K_1(0,0)}{\I k_1} + o\left(\frac{1}{k}\right);\quad \phi_1^\prime(k,0) = -\I k_1 +K_1(0,0) + o(1),\\
\phi(k,0) &=1 -\frac{K(0,0)}{\I k} + o\left(\frac{1}{k}\right);\quad \phi^\prime(k,0,t) = \I k - K(0,0) + o(1).
\end{align*}
Since $\frac{k_1}{k}=1+ O\left(\frac{1}{k^2}\right)$ we further infer
\begin{align}\nn
W(k,t)&=\I k -\frac{k}{k_1}K_1(0,0,t) - K(0,0) + \I k_1 - \frac{k_1}{k} K(0,0)\\ \nn
& -K_1(0,0) +o(1)=2\I k - 2(K+K_1)(0,0) + o(1).
\end{align}
Thus,
\beq\label{asT}
T(k)=1 + \frac{K(0,0) + K_1(0,0)}{\I k} +o\left(\frac{1}{k}\right).
\eeq
Next, since $k_1 = k+\frac{c^2}{2k} +O\left(\frac{1}{k^2}\right)$,
\begin{align}\label{as8}
\phi_1(k,x)\E^{\I k x} = &\E^{\I (k - k_1) x}\left(1 -\frac{K_1(x,x)}{\I k_1}) + o\left(\frac{1}{k}\right)\right)=\\ \nn
=& \E^{-\frac{2c^2}{2 k}(1+ o(1))x}\left(1 -\frac{K_1(x,x)}{\I k}\right) +o\left(\frac{1}{k}\right)= \left(1+\frac{c^2 x}{2\I k}\right)\times\\ \nn
& \times\left(1 -\frac{K_1(x,x)}{\I k}\right)+o\left(\frac{1}{k}\right) = 1 +\frac{\frac{c^2}{2} x -
K_1(x,x)}{\I k} + o\left(\frac{1}{k}\right).
\end{align}
From formulas \eqref{1} it follows, that
\[
\frac{d}{dx}(K_1(x,x) + K(x,x))=\frac{c^2}{2}
\]
and therefore
\beq\label{iden5}
K_1(x,x) + K(x,x)=\frac{c^2 x}{2} + K_1(0,0) + K(0,0).
\eeq
Combining \eqref{asT}--\eqref{iden5} we get
\[
T(k)\phi_1(k,x)\E^{\I k x}= 1 +\frac{K(x,x)}{\I k} + o\left(\frac{1}{k}\right).
\]
On the other side,
\[
\phi(k,x)\E^{-\I k x}=1 -\frac{K(x,x)}{\I k} + o\left(\frac{1}{k}\right),
\]
which finishes the proof.
\end{proof}

\begin{proof}[Proof of Theorem \ref{thm:vecrhp}]
We begin by checking that the jump condition for the vector $m(k,x,t)$ defined in \eqref{defm} has the form \eqref{eq:jumpcond}. Since our further considerations are mostly algebraically, we omit the variables $x,t$ and sometimes also $k$ in notations whenever possible.

Consider $k\in\Sigma=\R$. Let $\left(\begin{smallmatrix}\alpha&\beta\\ \gamma & \delta\end{smallmatrix}\right)$ be the unknown jump matrix. Since $T(-k)=\ol {T(k)}$, $\phi_1(-k)=\ol{\phi_1(k)}$ for $k\in\Sigma$, the entries of $m$ satisfy
\[
T\phi_1\E^{\I k x}  = \ol\phi \, \E^{\I k x}\alpha + \ol {T\phi_1} \E^{-\I k x}\gamma,\qquad
\phi\, \E^{-\I k x}=\ol\phi\, \E^{\I k x}\beta + \ol{ T\phi_1}\E^{-\I k x}\delta.
\]
Multiply the first equality by $\E^{-\I k x}$, the second one by $\E^{\I k x}$, and then conjugate both of them. Abbreviating
\beq\label{pereh}
\gamma^\prime=\gamma\E^{-2\I k x} \ \mbox{ and}\ \beta^\prime=\beta\E^{2\I k x},
\eeq
we finally get
\[
\ol\alpha\phi=\ol{T \phi_1} -T\ol{\gamma^\prime}\phi_1,\qquad T\ol\delta\phi_1  = \ol\phi -\ol{\beta^\prime}\phi.
\]
Now divide the first by $\ol T$ and compare both with \eqref{rscat}. This shows $\delta=1$, $-\ol{\beta^\prime}=R(k,t)$,
$\ol{\alpha}=T_1\ol T$, and $-\ol{\gamma^\prime} \frac{T}{\ol T}=R_1$. Applying \eqref{pereh}, \eqref{realos} and the evolution formula for $R(k,t)$ from Lemma~\ref{lemsc}, {\bf 5.}, we get $v(k)$ for $k\in\Sigma$.

Now let $k\in\Sigma_c^U\subset\C^+$ implying that we have to work with the upper case in \eqref{defm}.  Since the function $\phi(k)$ is real-valued and has no jump on this set, the equations for the entries of the jump matrix read as follows (assuming $k\in [0,\I c]_+$)
\[
T\phi_1\E^{\I k x}=\ol{T\phi_1}\E^{\I k x}\alpha + \phi\E^{-\I k x}\gamma,\qquad \phi\E^{-\I k x}\gamma=
\ol{T\phi_1}\E^{\I k x}\beta + \phi\E^{-\I k x}.
\]
From the last equality $\beta=0$, $\delta=1$. Abbreviating $\gamma\E^{-2\I k x}=:\gamma^\prime$ and divide this equality by $-\ol T$. By virtue of \eqref{posl}
we get
\[
R_1\phi_1+\ol{\phi_1}\alpha=\frac{-\gamma^\prime}{\ol T}\phi.
\]
Comparing this equality with the first of the scattering relations \eqref{rscat} we obtain $\alpha=1$ and $\gamma=-T_1(k,t)\ol{T(k,t)}\E^{2\I k x}$. Since
$-T_1\ol T=-\frac{k}{k_1}|T_1|^2$ for $k\in\Sigma_c$, the corresponding formula from item {\bf 4.} of Lemma \ref{lemsc} establishes the formula for $v(k)$, $k\in\Sigma^U_c$. For $k\in\Sigma^L_c$ we use property \eqref{proptau} and formula $R(k)=\ol{R(-k)}$ valid for
$-C_0\im k<0$.
\end{proof}

\bigskip
\noindent{\bf Acknowledgments.}
We thank Alexander Minakov for useful discussions and Aelxei Rybkin for pointing out several misprints in an earlier version.

\end{document}